\documentclass[11pt,reqno]{amsart}
\usepackage{amsaddr}

\usepackage[a4paper,
			left=2.5cm, 
			right=2.5cm, 
			top=2.5cm]{geometry}

\usepackage{amsmath,amssymb,amsthm,amsfonts,amsxtra,bbm}

\usepackage{tikz-cd}

\usepackage[hyperfootnotes=true,
			pdffitwindow=true,
			plainpages=false,
			pdfpagelabels=true,
			pdfpagemode=UseOutlines,
			pdfpagelayout=SinglePage,
			hyperindex,]{hyperref}

\numberwithin{equation}{section}

\usepackage[shortlabels]{enumitem}

\theoremstyle{plain}
\newtheorem{theorem}{Theorem}
\newtheorem{lemma}[theorem]{Lemma}
\newtheorem{corollary}[theorem]{Corollary}
\newtheorem{proposition}[theorem]{Proposition}
\theoremstyle{definition}
\newtheorem{definition}{Definition}
\newtheorem{example}{Example}
\theoremstyle{remark}
\newtheorem{remark}{Remark}

\newcommand{\Z}{\mathbb Z}
\newcommand{\R}{\mathbb R}
\newcommand{\C}{\mathbb C}

\renewcommand{\O}{\mathrm O}
\newcommand{\U}{\mathrm U}
\newcommand{\SO}{\mathrm{SO}}

\newcommand{\USp}{\mathrm{USp}}
\newcommand{\herm}{\mathrm{Herm}}
\newcommand{\Mat}{\mathrm{Mat}}
\newcommand{\GL}{\mathrm{GL}}

\newcommand{\dv}{\mathrm{d}}
\newcommand{\tr}{\mathrm{Tr}}
\newcommand{\diag}{\mathrm{diag}}
\renewcommand{\o}{\mathfrak o}
\newcommand{\usp}{\mathfrak {usp}}
\newcommand{\s}{\mathcal S}
\newcommand{\f}{\mathcal F}
\newcommand{\h}{\mathcal H}
\newcommand{\m}{\mathcal M}
\newcommand{\sgn}{\mathrm{sgn}}

\begin{document}
	
\title{Derivative principles for invariant ensembles}

\author{Mario Kieburg}
\email{m.kieburg@unimelb.edu.au (M.Kieburg)}

\address{School of Mathematics and Statistics, University of Melbourne, 813 Swanston Street, Parkville, Melbourne VIC 3010, Australia}

\author{Jiyuan Zhang}
\email{jiyuanzhang.ms@gmail.com (J. Zhang)}

\address{Department of Mathematics, Katholieke Universiteit Leuven, Celestijnenlaan 200B, Leuven B-3001, Blegium}




\begin{abstract}
In the present work we show that the joint probability distribution of the eigenvalues can be expressed in terms of a differential operator acting on the distribution of some other matrix quantities. Those quantities might be the diagonal or pseudo-diagonal entries as it is the case for Hermitian matrices. These representations are called derivative principles. We show them for the additive spaces of the Hermitian, Hermitian antisymmetric, Hermitian anti-self-dual, and complex rectangular matrices as well as for the two multiplicative matrix spaces of the positive definite Hermitian matrices and of the unitary matrices. In all six cases we prove the uniqueness of the derivative principles.
\end{abstract}

\maketitle

\


\section{Introduction}\label{s1}

In elementary probability theory, it is well known that for two independent real random vectors $x,y\in\R^n$ associated with probability distributions $f(x)$ and $g(y)$ respectively, their sum $x+y$ is a random vector with distribution equal to the additive convolution of $f$ and $g$, i.e.,
\begin{equation}\label{1.0.1}
f\ast g(z):=\int_{\R^n} f(z-t)g(t)\dv t.
\end{equation}
This result can be derived using multivariate harmonic analysis in terms of characteristic functions in the language of probability theory.

The study of the matrix analogue has drawn recent attention in terms of finding the eigenvalues of $A+B$ where $A$ and $B$ are Hermitian matrices. One branch of this is to consider $A$ and $B$ being fixed, which is called the Horn problem~\cite{Ho62}. We refer the rich literatures for its various connections with representation theory~\cite{CZ17,Fu00,Kl98}, combinatorics~\cite{KTT06,KT01}, and algebraic geometry~\cite{Kn99}. Another branch considers $A$ and $B$ drawn from specific ensembles which are invariant under conjugation of respective matrix groups---these ensembles only depend on their eigenvalue distributions while their eigenvectors become independent and are Haar distributed. For such eigenvalues to be fixed, i.e. the eigenvalue distributions are Dirac deltas, the sum of those matrices generalises Horn's problem to a sum of randomized orbits created by the conjugate action of $\O(n)$ and $\U(n)$, see~\cite{FG06,Zu17,Fa19,FZ19}. We would like to draw also attention to some recent works discussing a multiplicative version of Horn's problem~\cite{FZ19,ZKF19}. For random eigenvalues, a big class of unitarily invariant ensembles, the P\'olya ensemble on $\herm(n)$, has been identified in~\cite{KR19,FKK17}. In those studies analogously to elementary probability theory, harmonic analysis on symmetric spaces appears to be an essential tool, e.g., see~\cite{KR19,FKK17,ZKF19}.

The essential difference between sums of real random variables and the eigenvalues of a sum of random matrices is that, the latter does not admit the simple convolution formula~\eqref{1.0.1}. In other words, the eigenvalues of the sum generally is not the same as the sum of eigenvalues, which is, $\mathrm{eig}(A+B)\ne \mathrm{eig}(A)+\mathrm{eig}(B)$ for $A,B$ being two Hermitian matrices. Our motivation for the present work, is to bypass such non-additivity by uniquely mapping eigenvalues to another set of random variables which admits the convolution formula~\eqref{1.0.1}. Such a relation, called the \textit{derivative principle}, was firstly found in~\cite{Fa06} in terms of Hermitian random matrices, which did not gain wide attention. It was then rediscovered in the framework of symplectic geometry in the study of a one-body quantum marginal problem~\cite{CDKW14}, and was later rephrased using again random matrix theory~\cite{MZB16}. One purpose of this manuscript is to draw more attention to the following result.

\begin{proposition}[Derivative principle~\cite{Fa06,CDKW14,MZB16}]\label{p1.1}
	Let $X\in\herm(n)$ be an $n\times n$ unitarily invariant random Hermitian matrix, i.e. its distribution is unchanged under any conjugation of $n\times n$ unitary matrices. Let $f$ denote its joint probability density of the eigenvalues, and $f_{\diag}$ denote its joint probability density of the diagonal entries. Under certain analytical assumptions (specified later in Proposition~\ref{p3.1.1}), the two admit the following relation:
	\begin{equation}\label{1.1}
	f(\mathbf x)=\frac{1}{\prod_{j=1}^nj!}\Delta(\mathbf x)\Delta(-\partial_{\mathbf x})f_{\diag}(\mathbf x),
	\end{equation} 
	where $\Delta(\mathbf x):=\prod_{1\le j<k\le n}(x_k-x_j)$ is the Vandermonde determinant and $\Delta(-\partial_{\mathbf x}):=\prod_{1\le j<k\le n}(\partial_{x_j}-\partial_{x_k})$ is a differential operator acting on $f_\diag$.
\end{proposition}

Proposition~\ref{p1.1} provides a bijection between the joint distributions of eigenvalues and diagonal entries. Now if $A$ and $B$ are drawn from two unitarily invariant ensembles, we know from $\diag(A+B)=\diag(A)+\diag(B)$ that their diagonal entries satisfy the convolution formula~\eqref{1.0.1}. Therefore, using Proposition~\ref{p1.1} we can recover the eigenvalues of $A+B$. Proofs of Proposition~\ref{p1.1} can be found in~\cite{Fa06,MZB16}---however, both of them are lacking discussions on the analytical requirements. Such assumptions will be added in Proposition~\ref{p3.1.1}, where we will prove it for the sake of completeness and also for a parallel comparison with proofs of the other analogous cases that we will introduce later.
 
As an application of Proposition~\ref{p1.1}, we would like to point out its potential in transferring results from classical probability theory to random matrices. Particularly, one sees that the vector of diagonal entries is a classical random vector and satisfies a classical central limit theorem (or the generalised version in terms of stable laws) in contrast to the one of the eigenvalues. Yet, with the help of the derivative principle, such a result can be directly carried over to Hermitian invariant ensembles, as well as results in terms of its rate of convergence; see~\cite{KZ21}. 

Coming from another perspective, one may encounter the P\'olya ensembles on $\herm(n)$ that have a very similar structure as the right hand side of~\eqref{1.1}, cf., Refs.~\cite{KR19,FKK17}. Indeed, one only needs to replace $f_\diag$ by a product of weight functions depending on individual $x_k$, and the differential operator $\Delta(-\partial_x)$ by $\Delta(-x\partial_x):=\prod_{j<k}(x_k\partial_{x_k}-x_j\partial_{x_j})$. In recent works~\cite{FKK17,Ki19,KFI19,KK16,KK19,ZKF19}, harmonic analysis for random matrices has been extended to sums of other matrix spaces as well as to products of random matrices with the help of other Harish-Chandra--Itzykson--Zuber-type integrals~\cite{HC57,IZ80} like the Berezin-Karpelevich integral~\cite{BK58,GW96} or the Gelfand-Na\u\i mark integral~\cite{GN57}. Interestingly, the corresponding P\'olya ensembles that arise out of all these computations have a distribution with a similar structure as the right side of~\eqref{1.1}. Therefore, we would like to answer the following natural question:
\begin{center}
	\it Can one generalise the derivative principle on unitarily invariant $\herm(n)$ ensembles \\to invariant ensembles for other types of matrices? 
\end{center}
 
The main purpose of the present work, is to prove the affirmation of this question with consideration of two types of invariant ensembles: additive and multiplicative ensembles. We consider the additive matrix spaces of Hermitian matrices $\herm(n)$ and of the Hankel class $M_\nu$. The latter comprises the two Lie algebras of the orthogonal group which are the real antisymmetric matrices $\o(n)$ and of the unitary symplectic group $\USp(2n)$ which are the quaternion anti-self-dual matrices $\usp(2n)$ and the complex rectangular matrices $\Mat(n,n+\nu)$. It was already pointed out in~\cite{FKK17} that they can be considered in a unifying way. Moreover, we study two sets of multiplicative ensembles on $\herm_+(n)$, set of positive Hermitian matrices, and the unitary group $\U(n)$. Any random matrix drawn from \textit{those} ensembles are invariant under their corresponding group actions.

For each of these ensembles, we will prove a derivative principle and show that the relation between the two quantities that are connected by this principle is unique. In the case of the additive matrix spaces, the matrix addition becomes a simple additive convolution even for the eigenvalues. Alas, this is not so simple for the multiplicative cases where we have been only able to derive the derivative principles.

For each of these ensembles, we will prove a derivative principle by giving an explicit formula analogous to~\eqref{1.1}, relating the eigenvalue/squared singular value distributions to multivariate weight functions. We will also show that the latter quantities are unique in such relations. Those formulae can potentially be used in future studies of sums and products of specific random matrix ensembles, such as in calculations of their eigenvalue statistics. Also with those explicit formulae, we are now able to explain why P\'olya ensembles in various matrix spaces (summarised in \cite{FKK17}) are indeed natural to consider---as allowing the multivariate weight function to be decomposable into a product of univariate functions is one of the most natural way to study determinantal processes which are compatible with matrix sums or products.

We organise the present work as follows. In section~\ref{s2}, we introduce our notations used throughout the work. Furthermore, we briefly review the harmonic analysis on multivariate Euclidean spaces as well as matrix spaces as they are the essential tools for proving our statements. In section~\ref{s3}, we first review the theory for additive invariant ensembles on $\herm(n)$, and then analogously develop similar statements for the additive invariant ensembles on $i\o(n)$, $i\usp(2n)$ and $\Mat(n,n+\nu)$. These discussions are carried over to multiplicative invariant ensembles on $\herm_+(n)$ and $\U(n)$ in section~\ref{s4}. We summarise our findings in section~\ref{s5} and in the two appendices we give more insights and proofs for some technical tools that have been used. We will also emphasise the close relationship between Lie algebras and all three additive ensembles in section~\ref{s3.4} with comparison to classical results in harmonic analysis, while such relation remains to be unravelled for multiplicative ensembles.
\vspace{-0.5em}
\section{Preliminaries}\label{s2}

\subsection{Invariant ensembles}\label{s2.1}

Before we come to our main statements we will briefly outline the basic properties of the considered matrix spaces, their corresponding invariant class of group actions, and their underlying measure. We also adapt notations used in~\cite{FKK17}. First let $\O(n), \U(n)$ and $\USp(2n)$ be the orthogonal, unitary and unitary symplectic group respectively. Their associated Lie algebras are denoted by $\o(n)$, $\mathfrak u(n)$ and $\usp(2n)$, respectively.

Due to the intimate relation of the space of Hermitian $n\times n$ matrices $\herm(n):=i\mathfrak u(n)$ we quite often consider those instead of the unitary Lie-algebra. This matrix space is endowed with a conjugation $\U(n)$-action $X\mapsto KXK^{-1}$. By spectral theory, under this action each $X\in \herm(n)$ can be diagonalised, i.e., $X=K\diag(x)K^{-1}$ for some $K\in \U(n)$ where $x=(x_1,\ldots,x_n)$ denotes the eigenvalues of $X$.

To assign probability measures to $\herm(n)$, we first define an underlying measure $\dv X$ for $X\in\herm(n)$ to be the Lebesgue product measure of the upper triangular entries, i.e.
\begin{equation}
\dv X:=\prod_{j=1}^n\dv x_{j,j}\prod_{1\le j<k\le n}\dv x_{j,k}^{(r)}\dv x_{j,k}^{(i)},
\end{equation}
where $x_{j,k}^{(r)}$ and $x_{j,k}^{(i)}$ are the real and imaginary part of the $(j,k)$ entry $x_{j,k}$. In random matrix theory we often consider probability measures absolutely continuous with respect to $\dv X$, and invariant under the endowed $\U(n)$-action. Those measures have a probability density $F\in L^1(\herm(n))$ (i.e., $F$ is absolutely integrable), which satisfies the relation $F(X)=F(KXK^{-1})$ for any $K\in \mathrm U(n)$ and $X\in\mathrm {Herm}(n)$. The collection of all those random variables $X$ is called the \textit{invariant ensemble on $\herm(n)$}.

To assign an invariant probability density $F(X)$, it is equivalent to choose its corresponding eigenvalue distributions $f(x)$ since the Haar measure describing the eigenvectors is unique. For invariant ensembles, there is a simple relation between $F$ and $f$ given by
\begin{equation}\label{1.1.2}
f(x)=\frac{\pi^{n(n-1)/2}}{\prod_{j=0}^nj!}\Delta(x)^2F(\diag(x)),
\end{equation}
where $\diag(x)$ is the diagonal matrix with $x=(x_1,\ldots,x_n)$ being its diagonal entries, and $\Delta(x):=\prod_{1\le j<k\le n}(x_k-x_j)$ is the Vandermonde product. We will use the above notations to discuss the derivative principle on $\herm(n)$ in section~\ref{s3.1}.

The second set of matrices we consider and denote by $\Mat(n,n+\nu)$ are the $n\times (n+\nu)$ rectangular complex matrices. In~\cite{FKK17}, the four matrix groups $\Mat(n,n+\nu), i\o(2n)$, $i\o(2n+1)$ and $i\usp(2n)$ are considered as a general class---\textit{the Hankel class} $M_\nu$, corresponding to a parameter $\nu$. We introduce them together as follows.

\begin{enumerate}
	\item The case $\nu\in\mathbb{N}_0$ corresponds to $\Mat(n,n+\nu)$, which is endowed with a left and right $\U(n)\times \U(n+\nu)$ action $X\mapsto K_1XK_2^{-1}$ for $K_1\in \U(n)$ and $K_2\in\U(n+\nu)$, and by a singular value decomposition, one has $X=K_1YK_2^{-1}$ with $Y=[\lambda_{j}\delta_{j,k}]_{\substack{j=1,\ldots,n\\k=1,\ldots,n+\nu}}$, where $\lambda_1,\ldots,\lambda_{n}$ are the singular values of $X$ and $\delta_{j,k}$ is the Kronecker delta. To understand these in terms of a conjugation action, the matrix space $\Mat(n,n+\nu)$ can be equivalently viewed as a $\herm(2n+\nu)$ chiral matrices with the following block form
	\begin{equation}\label{1.1.4}
	\begin{bmatrix}
	0&X\\X^\dagger&0
	\end{bmatrix},\quad X\in\Mat(n,n+\nu)
	\end{equation}
	with $X^\dagger$ the Hermitian adjoint of $X$. The action is then a conjugation $\U(n)\times \U(n+\nu)$ action along
	\begin{equation}\label{1.1.5}
	\begin{bmatrix}
	0&X\\X^\dagger&0
	\end{bmatrix}=K\,\begin{bmatrix}
	0&Y\\Y^\dagger&0
	\end{bmatrix}\,K^{-1},\quad K=\begin{bmatrix}K_1&0\\0&K_2\end{bmatrix}\in \U(n)\times \U(n+\nu).
	\end{equation}
	
	\item The set $i\o(m)$ denotes the set of $m\times m$ imaginary anti-symmetric matrices which is equivalent to the Lie-algebra of the orthogonal matrices. We identify $\nu=-1/2$ with $i\o(2n)$ and $\nu=1/2$ with $i\o(2n+1)$, for $n\in\mathbb{N}$. They are endowed with a conjugation $\O(2n)$-action and $\O(2n+1)$-action respectively, and can be block-diagonalised as follows
	\begin{equation}\label{1.1.6}
	X=K\,\diag\left(\begin{bmatrix}
	0&i\lambda_1\\-i\lambda_1&0
	\end{bmatrix},\ldots,\begin{bmatrix}
	0&i\lambda_n\\-i\lambda_n&0
	\end{bmatrix}\right)\,K^{-1},\quad K\in\O(2n),
	\end{equation}
	\begin{equation}\label{1.1.7}
	X=K\,\diag\left(\begin{bmatrix}
	0&i\lambda_1\\-i\lambda_1&0
	\end{bmatrix},\ldots,\begin{bmatrix}
	0&i\lambda_n\\-i\lambda_n&0
	\end{bmatrix},0\right)\,K^{-1},\quad K\in\O(2n+1),
	\end{equation}
	respectively. Here $\pm \lambda_1,\ldots,\pm \lambda_n\in\R$ are the pairs of eigenvalues of $X$.
	
	\item The symmetric matrix space $i \usp(2n)$ denotes the set of $2n \times 2n$ complex Hermitian matrices $[q_{j,k}]_{j,k=1}^n$, with $2\times 2$ quaternionic sub-block structure
	\begin{equation}\label{1.1.8}q_{j,k}=\begin{bmatrix}z&w\\-\bar{w}&\bar{z}\end{bmatrix},\quad z,w\in\C.\end{equation}
	It also corresponds to $\nu=1/2$ as the joint probability density of the non-zero eigenvalues cannot distinguish between $i \usp(2n)$ and $i \o(2n+1)$. The conjugate $\USp(2n)$-action gives a diagonal matrix of pairs of its eigenvalues $\pm \lambda_1,\ldots,\pm \lambda_n\in\R$
	\begin{equation}\label{1.1.9}
	X=K\,\diag(\lambda_1,-\lambda_1,\ldots,\lambda_n,-\lambda_n)\,K^{-1},\quad K\in\USp(2n).
	\end{equation}
\end{enumerate}

For a unified discussion of the above cases, we collect the invariance groups ($\U(n)\times \U(n+\nu), \O(2n), \O(2n+1)$ and $\USp(2n)$, respectively) in a single symbol $K_\nu$ with $\nu\in\mathbb{N}_0\cup\{\pm 1/2\}$. It is convenient to consider the squared singular values $x=(x_1,\ldots,x_n)=(\lambda_1^2,\ldots,\lambda_n^2)$ of $X$ instead of the eigenvalues. We use a unified notation $\iota(x)$ to denote an embedding of $\sqrt{x_1}=\lambda_1,\ldots,\sqrt{x_n}=\lambda_n$ into the matrix spaces. That is, $\iota$ defines a map from $\R^N$ to the respective maximal Abelian subspace, which can be read off in~\eqref{1.1.5},~\eqref{1.1.6}~\eqref{1.1.7}, and~\eqref{1.1.9}, respectively. All $M_\nu$ matrix decompositions can then be written as $X=K\iota(x)K^{-1}$ for $K\in K_\nu$.

The underlying reference measure for $\Mat(n,n+\nu), i\o(2n)$ and $i\o(2n+1)$ are the Lebesgue product measures $\dv X:=\prod_{(j,k)\in A}\dv x_{j,k}$, with the index set $A$ specified in each case by
\begin{equation}
A=\begin{cases}
\{(j,k):j=1,\ldots,n, k=1,\ldots,n+\nu\},&M_\nu=\Mat(n,n+\nu),\\
\{(j,k): 1\le j<k\le 2n\},& M_\nu=i\o(2n),\\
\{(j,k): 1\le j<k\le 2n+1\},& M_\nu=i\o(2n+1),\\
\{(j,k): 1\le j<k\le n\},&M_\nu=i\usp(2n).
\end{cases}
\end{equation}
For $\Mat(n+\nu)$, $\dv x_{j,k}$ is the Lebesgue measure on $\C$ while for $i\o(2n)$ and $i\o(2n+1)$ it is the one on $\R$. For $i\usp(2n)$, $\dv x_{j,k}$ represents the Lebesgue measure on the quaternions.

Summarising the above, a probability density on $M_\nu$ is then a positive $L^1$-function $F(X)$ with $F(X)=F(KXK^{-1})$. Similarly to~\eqref{1.1.2}, we have relations between $F(X)$ and the joint probability density $f(x)$ of its squared singular value distribution $x$; it is
\begin{equation}\label{1.1.10}
	\begin{split}
f(x)&=\frac{\pi^{n(n+\nu)}}{n!C_\nu}\Delta(x)^2\prod_{j=1}^nx_j^{\nu}\,F(\iota(x)),\\ C_\nu&=\begin{cases}2^{n(n-1)}\prod_{j=0}^{n-1}{j!\,\Gamma(j+\nu+1)},&M_{1/2}=i\usp(2n),\\
\prod_{j=0}^{n-1}{j!\,\Gamma(j+\nu+1)},&\text{otherwise,}
\end{cases}
\end{split}
\end{equation}
(see e.g.~\cite{FKK17}). The collection of such random matrices is the set of \textit{invariant ensemble on $M_\nu$}.

We will use these notions to derive derivative principles on each case of the Hankel class in section~\ref{s3.2} and~\ref{s3.3}. For this purpose, we also introduce a notion---the \textit{pseudo-diagonal entries}, which are real entries that play a similar role as the diagonal entries in $\herm(n)$ matrices. In $\Mat(n,n+\nu)$, they represent the real parts of the matrix entries $x_{1,1}, x_{2,2}, \ldots,x_{n,n}$. In both $i\o(2n)$ and $i\o(2n+1)$, they are the matrix entries $x_{1,2},x_{3,4},\ldots,x_{2n-1,2n}$, all divided by the overall factor $i$ so that they are real. In $i\usp(2n)$, they are given by the diagonal entries $x_{2,2},x_{4,4},\ldots,x_{2n,2n}$ as they coincide up to a sign with the entries $x_{1,1},x_{3,3}\ldots x_{2n-1,2n-1}$, because of the quaternion $2\times 2$ block structure.

In the following sections we will identify both $\herm(n)$ and $M_\nu$ as additive matrix groups, equipped with matrix addition. Besides, in the present work we also introduce $\herm_+(n)$, the set of all $n\times n$ invertible matrices, and the unitary group $\U(n)$, as other two multiplicative matrix spaces. 

The endowed group invariance and the reference measure on $\herm_+(n)$ are the natural restrictions of those in $\herm(n)$. We consider \textit{invariant ensemble on $\herm_+(n)$} as a subclass of invariant ensembles on $\herm(n)$, provided their eigenvalues are positive, and a Weyl decomposition formula~\eqref{1.1.2} holds. In terms of multiplication, we consider the following Hermitised product: $(A,B)\mapsto A^{1/2}BA^{1/2}$. This is not a group operation, as associativity is not satisfied, but it is related by the matrix multiplication of $\GL(n)$, $n\times n$ invertible complex matrices. Consider two $\GL(n)$ matrices $G_1$, $G_2$. They form two $\herm_+(n)$ matrices $X_1=G_1^\dagger G_1$ and $X_2=G_2^\dagger G_2$. From the identity
\begin{equation}
\det(\lambda I_n-X_1^{1/2}X_2X_1^{1/2})=\det(\lambda I_n-(G_2G_1)^{\dagger}G_2G_1),
\end{equation}
one can see that the squared singular values of $G_2G_1$ are equal to the eigenvalues of $X_1^{1/2}X_2X_1^{1/2}$. Such an operation on invariant ensembles on $\herm_+(n)$ motivates a multiplicative derivative principle for $\herm_+(n)$, which will be introduced in section~\ref{s4.1}.

The unitary group $\U(n)$, equipped with the usual matrix multiplication, is also endowed with a conjugation $\U(n)$-action. As its reference measure we choose the normalised and uniquely given Haar measure $\mu(dK)$. We therefore consider probability measures absolutely integrable with respect to $\mu(dK)$, that are also invariant under the conjugate group action of $\U(N)$. Such random matrices give an \textit{invariant ensemble on $\U(n)$} with a probability density satisfying $F(X)=F(KXK^{-1})$ for any $X,K\in \U(n)$. Its eigenvalue distribution $f(x)$ is a function on the $n$ dimensional torus since the eigenvalues are distributed on the unit circle on the complex plane. The relation between $F$ and $f$ is given by
\begin{equation}\label{2.1.12}
f(x)=\frac{1}{(2\pi)^nn!}|\Delta(x)|^2F(\diag(x)),
\end{equation}
(see e.g.~\cite[\S 3]{Fo10}). We will derive a derivative principle for $\U(n)$ in section~\ref{s4.2}.

\subsection{Multivariate transforms}\label{s2.2}

In this subsection, we list all multivariate transforms relevant for our discussions, which are analogues to the matrix version spherical transform introduced in the next subsection. We will begin by giving our convention of the Fourier transform on $\R^n$.

\begin{definition}[Fourier transform]
	For a function $\tilde{f}\in L^1(\R^n)$, its multivariate Fourier transform is given by
	\begin{equation}\label{Fourier}
	\f \tilde{f}(s):=\int_{\R^n}\tilde{f}(x)\prod_{j=1}^ne^{ix_js_j}\,\dv x,\quad s\in\R^n.
	\end{equation}
	It has the following well-known properties
	\begin{enumerate}[(1)]
		\item \textit{Inversion}. The injectivity of $\f$ allows an inverse transform after proper restriction to the image $\f (L^1(\R^n))$, which is given by
		\begin{equation}
		\f^{-1}(\f \tilde{f})(x):=\lim_{\varepsilon \to0^+}\frac{1}{(2\pi)^n}\int_{\R}\f \tilde{f}(s)\prod_{j=1}^ne^{-ix_js_j-\varepsilon  s_j^2}\dv x.
		\end{equation}
		The regularisation $\prod_{j=1}^ne^{-\varepsilon  s_j^2}$ in $\varepsilon>0$ can be dropped when $\f \tilde{f}$ is absolutely integrable, too.
		\item \textit{Convolution}. The Fourier transform of an additive convolution of two functions is
		\begin{equation}\label{2.2.2}
		\f (\tilde{f}_1\ast \tilde{f}_2)=\f \tilde{f}_1\cdot \f \tilde{f}_2\quad\text{ with }\quad (\tilde{f}_1\ast \tilde{f}_2)(x):=\int_{\R^n} \tilde{f}_1(x-y)\tilde{f}_2(y)\dv y,
		\end{equation}
		where $\ast$ refers to the additive convolution, and $x-y$ refers to the entry-wise subtraction.
		\item \textit{Differentiation}. For $\tilde{f}\in L^{1,m}_\f(\R^n)$, one has
		\begin{equation}\label{2.2.4b}
		\f\Big[(i\partial_{x_j})^k \tilde{f}(x)\Big](s)=s_j^k\f \tilde{f}(s),\quad 0\le k\le m.
		\end{equation}
	\end{enumerate}
\end{definition}

In part (3), the underlying set $L^{1,m}_\f(\R^n)$ is given by 
\begin{equation}\label{2.2.4a}
L^{1,m}_{\f}(\R^n):=\bigg\{ \tilde{f}\in L^1(\R^n): \int_{\R^n} \left|x^a\partial_x^{b}\tilde{f}(x) \right|\dv x \le \infty, \forall a,b\in\mathbb{N}_0^n\ {\rm and}\ |a|,|b|\leq m\bigg\}.
\end{equation}
Here $a, b$ are multi-indices, i.e. $x^a=x_1^{a_1}x_2^{a_2}\ldots x_n^{a_n}$ and $|a|=a_1+a_2+\ldots+a_n$. Functions in this set are nice enough that even its Fourier transform is differentiable as well as sufficiently integrable for dropping the Gaussian regularisation in the inverse transform. Pushing $m$ to infinity gives the well-known Schwartz functions. We also remark here that such differentiability can be relaxed with the help of distribution theory, but for convenience we keep this differentiable version.

For the study of the Hankel class of matrix spaces $M_\nu$, another additive multivariate transform needs to be introduced -- the Hankel transform, which is the reason for the chosen name \textit{Hankel class}. Let us remark that the joint probability density of the squared singular values $f$ acts only on the $\mathbb{R}_+^n$, for which the Hankel transform is also defined.

\begin{definition}[Hankel transform]
	Let $\nu$ be a parameter in the set $\mathbb{N}_0\cup\{\pm 1/2\}$. For a function $\hat{f}\in L^1(\R_+^n)$, its Hankel transform with parameter $\nu$ and its inverse transform are given by
	\begin{equation}
	\begin{split}\label{2.2.8}
	\h_\nu \hat{f}(s)&=\int_{\R_+^n} \dv x\,\hat{f}(x)\prod_{j=1}^nJ_\nu(2\sqrt{x_js_j})(x_js_j)^{-\nu/2},\\
	\h_\nu^{-1}(\h_\nu \hat{f})(x)&=\lim_{\varepsilon \to^+0}\int_{\R_+^n} \dv s\,\h_\nu \hat{f}(s)\prod_{j=1}^nJ_\nu(2\sqrt{x_js_j})(x_js_j)^{\nu/2}e^{-\varepsilon  s_j},
	\end{split}
	\end{equation}
	where $J_\nu$ is the Bessel function of the first kind with parameter $\nu$. The regularisation $\prod_{j=1}^ne^{-\varepsilon s_j}$ in $\varepsilon>0$ can be again dropped when $\h_\nu \hat{f}(s)\prod_{j=1}^n (1+s_j)^{\nu/2-1/4}$ is absolutely integrable.
\end{definition}

The Hankel transform also satisfies a differentiation formula if the function $\hat{f}$ satisfies some boundary conditions so that the integration by parts can be done without producing additional terms, cf.~\cite[Definition 2.3]{FKK17}. For instance, choosing the function $\hat{f}(x)=x^\nu \tilde{f}(\sqrt{x})$ with an even (in each of the $n$ arguments) function $\tilde{f} \in L_\f^{1,2m}(\R^n)$, we have
\begin{equation}\label{2.2.5b}
\h_\nu[(-{x_j}^\nu\partial_{x_j}x_j^{1-\nu}\partial_{x_j})^k\hat{f}(x)](s)=s_j^k\h_\nu \hat{f}(s),\quad 0\le k\le m.
\end{equation}
This is simply due to the fact that $-{x}_j^\nu\partial_{x_j}x_j^{1-\nu}\partial_{x_j}(J_\nu(2\sqrt{x_js_j})(x_js_j)^{\nu/2})=s_jJ_\nu(2\sqrt{x_js_j})(x_js_j)^{\nu/2}$. The condition of $\tilde{f} \in L_\f^{1,2m}(\R^n)$and $\hat{f}(x)=x^\nu \tilde{f}(\sqrt{x})$ are indeed enough as can be readily checked when rewriting the Hankle transform as follows
\begin{equation}
\h_\nu \hat{f}(s)=\int_{\R^n} \dv \lambda\,\left(\prod_{j=1}^n\lambda_j^{2\nu}|\lambda_j|\right)\, \tilde{f}(\lambda)\prod_{j=1}^nJ_\nu(2\lambda_j\sqrt{s_j})(\lambda_j^2s_j)^{-\nu/2}.
\end{equation}
Then the differential operator becomes $(-\lambda_{j}^{2\nu-1}\partial_{\lambda_j}\lambda_j^{-2\nu}\partial_{\lambda_j}/4)^k$. The factor $x^\nu$ in front of $\tilde{f}(\sqrt{x})$ is reminiscent to the factor appearing in~\eqref{1.1.10}.

Despite the correspondence of the differentiation formula of the Fourier and Hankel transform, $\h_\nu$ does not allow a simply additive convolution formula~\eqref{2.2.2} unlike the Fourier transform. Indeed, it is well known that it corresponds to an additive convolution on space of real vectors of dimension $2\nu+2$.

Later in section~\ref{s3.3}, we will make use of a combination of $\h_\nu^{-1}$ and $\f$, which gives an integral transform called inverse Abel transform. It is a transform that describes the relation between spherical harmonics to plane waves and is usually employed for $\nu=0$ and $\nu=1/2$ in practice since they correspond to two and three dimensional analysis, e.g., see~~\cite[\S 8]{Po96}.
In literatures, there is also a very general definition for the inverse Abel transform in terms of spherical transforms~\cite{He84}. Thus, we will attach a proof of the following proposition in Appendix~\ref{aa}. As a side-remark, both, the inverse Hankel as well as the Fourier transform, are bijections in the space $L_\f^{1,2m}(\R^n)$, therefore, the inverse Abel transform is invertible, too.

\begin{proposition}[Inverse Abel transform]\label{p2.2.2}
	Let $\nu\in\mathbb{N}_0$ and $\tilde{f} \in L_\f^{1,\nu+1}(\R^n)$ be symmetric in each of its arguments, as well as $[\prod_{j=1}^n (x_j^{-1}\partial_{x_j})^{\nu+1}f]\in L^1(\R^n)$. Then, the inverse multivariate Abel transform is given by a combination of an inverse Hankel transform and a Fourier transform with the explicit integral expression
	\begin{equation}\label{2.2.8a}
	\mathcal{A}_\nu^{-1}\tilde{f}(x):=\h^{-1}_\nu [\f \tilde{f}(2\sqrt{s})](x)=\prod_{j=1}^nx_j^\nu\int_{x_1}^\infty\dv y_1\ldots\int_{x_n}^\infty\dv y_n\left(\prod_{j=1}^n\frac{1}{\sqrt{y_j-x_j}}\partial_{y_j}^{\nu+1}\right)\tilde{f}(\sqrt{y_1},\ldots,\sqrt{y_n}),
	\end{equation}
	for almost all $x\in\mathbb{R}_+^n$.
\end{proposition}

The function $\tilde{f}$ will be related to the probability $F$ on the matrix space $M_\nu$ but should not be confused with the joint probability density $f$ of the squared singular values $x$ of the random matrix $X$. The same also holds for the other transforms such as the Fourier transform we have introduced.

As a third and fourth transform, we introduce the multivariate Mellin transform and Fourier series that will be extremely helpful in dealing with the multiplicative convolutions. They are the multivariate analogues of the spherical transforms on $\herm_+(n)$ and $\U(n)$, respectively.

\begin{definition}[Mellin transform]
	 The multivariate transform Mellin transform of a function $\tilde{f}\in L^1(\R_+^n)$ is given by
	\begin{equation}
	\m \tilde{f}(s):=\int_{\R^n_+}\tilde{f}(x)\prod_{j=1}^nx_j^{s_j-1}\,\dv x,\quad s\in ((a,b)\times i\R)^n.\,
	\end{equation}
	where $(a,b)\times i\R$ is the fundamental strip for each variable with $-\infty\leq a\leq1\leq b\leq\infty$.
\end{definition}

The Mellin transform has the following three properties:
	\begin{enumerate}[(1)]
		\item \textit{Inversion}. The injectivityof  $\m$ allows an inverse transform on its image, which is given by
		\begin{equation}
		\m^{-1}[\m \tilde{f}](x):=\lim_{\varepsilon \to0^+}\frac{1}{(2\pi)^n}\int_{\R^n}\m \tilde{f}(c+is)\prod_{j=1}^nx_j^{-c-is_j}e^{-\varepsilon  s_j^2}\,\dv s,
		\end{equation}
		where $c$ can be any fixed real number located in the fundamental strip. The regularisation $\prod_{j=1}^Ne^{-\varepsilon s_j^2}$ in $\varepsilon >0$ can be dropped if $\m \tilde{f}(c+is)$ is absolutely integrable in $s\in \mathbb{R}$ for some suitable $c\in(a,b)$.
		\item \textit{Convolution}. The Mellin transform of the multiplicative convolution of two functions is equal to the multiplication of their Mellin transforms, i.e.
		\begin{equation}
		\m (\tilde{f}_1\circledast \tilde{f}_2)=\m \tilde{f}_1\cdot \m \tilde{f}_2,\quad (\tilde{f}_2\circledast \tilde{f}_2)(x):=\int_{\R^n} \tilde{f}_1(x/y)\tilde{f}_2(y)\dv y,
		\end{equation}
		where $\circledast$ refers to the multiplicative convolution, and $x/y$ refers to the entry-wise division.
		\item \textit{Differentiation}. For $f\in L^{1,m}_\m(\R_+^n)$, one has
		\begin{equation}
		\m\Big((-x_j\partial_{x_j})^k \tilde{f}(x)\Big)(s)=s_j^k\m \tilde{f}(s),\quad 0\le k\le m.
		\end{equation}
	\end{enumerate}

The set $ L^{1,m}_\m(\R_+^n)$ is given by
\begin{equation}
L^{1,m}_{\m}(\R_+^n):=\bigg\{ \tilde{f}\in C^{m}(\R_+^n): \int_{\R_+^n} \left|x^a(x\partial_x)^{b}\tilde{f}(x) \right|\dv x \le \infty, \forall a,b\in \mathbb{N}_0^n\text{ and }|a|,|b|=0,\ldots,m\bigg\}.
\end{equation}
Here $a,b$ are anew multi-indices. Similarly to $L^{1,m}_\f(\R^n)$, this set allows up to $m$ times the application of $(-x_j\partial_{x_j})$ on $\tilde{f}$, and the Mellin transforms obtained are all integrable.

Eventually, we come to the Fourier series which plays an important role for the multiplicative convolutions on the unitary group $\U(n)$. Those arise because the eigenvalues $x$ of a unitary matrix live on the $n$-dimensional torus.

\begin{definition}[Fourier series]
	Let $\tilde{f}\in L^1((-\pi,\pi]^n)$ and bounded, which is understood as $2\pi$-periodic piece-wise continuous function on $\R$. Then, its multivariate Fourier series and inversion are given by
	\begin{equation}\label{2.2.16}
	\begin{split}
	\f \tilde{f}(s):&=\int_{(-\pi,\pi]^n}\tilde{f}(x)\prod_{j=1}^ne^{ix_js_j}\,\dv x,\\ \f^{-1}(\f \tilde{f})(x):&=\lim_{\varepsilon \to0^+}\frac{1}{(2\pi)^n}\sum_{s\in\Z^n}\f \tilde{f}(s)\prod_{j=1}^ne^{-ix_js_j-\varepsilon  s_j^2}.
	\end{split}
	\end{equation}
	The regularisation $\prod_{j=1}^Ne^{-\varepsilon s_j^2}$ in $\varepsilon >0$ can be dropped when $\f \tilde{f}\in l^1(\mathbb{Z}^n)$. We slightly abuse the notation $\f$, but from the context one can distinguish it from the Fourier transform.
\end{definition}

Fourier series also satisfy a convolution property similar to~\eqref{2.2.2} with $(-\pi,\pi]^n$ replacing $\R^n$ in the integral and a differentiation property given exactly by~\eqref{2.2.4b}. The set of functions for the differentiation property is then denoted by
\begin{equation}\label{2.2.16b}
L^{1,m}_{\f}((-\pi,\pi]^n):=\bigg\{ \tilde{f}\in L^1((-\pi,\pi]^n): \int_{(-\pi,\pi]^n} \left|x^a\partial_x^{b}\tilde{f}(x) \right|\dv x \le \infty, \forall a,b\in\mathbb{N}_0^n\ {\rm and}\ |a|,|b|\leq m\bigg\}.
\end{equation}
Here, the differentiation has to hold at the boundary $\pm\pi$, too, since the functions are considered to be $2\pi$ periodic in each of its $n$ arguments. This is very natural as they can be originally understood as functions on an $n$-dimensional torus.

\subsection{Spherical transforms}

The author of~\cite{He84} provided a theoretic framework for a generalisation of the Fourier transform on all locally compact Lie groups with a group invariance which is called the \textit{spherical transform}. Roughly speaking this theory generalises the above ideas, allowing us to discuss operations of invariant random variables on a Lie group. One example is $\herm(n)$, considered to be the additive matrix group.

In practice, the spherical transform for invariant ensembles on $\herm(n)$ has been repeatedly used in the context of random matrix theory, which is the Harish-Chandra--Itzykson-Zuber (HCIZ) integral~\cite{HC57,IZ80}. Choosing two real vectors $x=(x_1,\ldots,x_n), s=(s_1,\ldots,s_n)\in\R^n$ with pairwise different components, the HCIZ integral is given by
\begin{equation}\label{2.2.3}
\int_{\U(n)}\mu(\dv K)\exp\left[i\tr\,K\diag(x)K^{-1} \diag(s)\right]=\left(\prod_{j=1}^{n-1}j!\right)\frac{\det[e^{ix_js_k}]_{j,k=1}^n}{\Delta(ix)\Delta(s)}=:\phi(x,s).
\end{equation}

When considering $\herm(n)$ as a matrix group equipped with matrix addition, one can define the spherical transform either by applying the HCIZ integral to a matrix Fourier transform, or by making use of the general framework introduced in~\cite{He84}. In comparison to the univariate case, a spherical transform opens up the possibility for analysing sums of invariant Hermitian random matrices.

\begin{definition}[Spherical transform on $\herm(n)$]
	For an invariant random matrix $X\in\herm(n)$ with distribution $F$ and joint probability density $f\in L^1(\R^n)$ of the eigenvalues $x$ of $X$, its spherical transform is defined by
	\begin{equation}\label{2.2.4}
	\s f(s):=\int_{\herm(n)}\dv X\,F(X)\exp(i\tr X\diag (s))=\int_{\R^n} \dv x\,f(x)\phi(x,s),
	\end{equation}
	for $s\in\R^n$ and $\phi$ given in~\eqref{2.2.3}. Its inverse transform is given by
	\begin{equation}\label{2.2.5}
	f(x)=\lim_{\varepsilon \to0^+}\frac{1}{\prod_{j=1}^{n}(j!)^2}\Delta(x)^2\int_{\R^n}\frac{\dv s}{(2\pi)^n}\,\s f(s)\phi(-x,s)\Delta(s)^2 \prod_{j=1}^ne^{-\varepsilon  s_j^2},
	\end{equation}
	(see e.g.~\cite{KR19}). The regularisation $\prod_{j=1}^ne^{-\varepsilon  s_j^2}$ in $\varepsilon >0$ is only necessary when the remaining integrand is not absolutely integrable.
\end{definition}

The benefit of this transformation shows when adding two invariant Hermitian random matrices $A$ and $B$, namely the convolution theorem, that is the analogue of~\eqref{2.2.2} with $\f$ replaced by $\s$, reads
	\begin{equation}\label{2.2.5a}
	f_{A+B}=\s^{-1}(\s f_A\cdot \s f_B).
	\end{equation}
We employ here a convenient notation $f_A$ to highlight that it is the joint probability distribution of the eigenvalues of the random matrix $A$.

The general framework of Helgason's textbook allows realisations of spherical transforms on other matrix groups, as well. In particular, for the Hankel class $M_\nu$, see~\cite{FKK17}, one can find group integrals corresponding to~\eqref{3.1.1}. For example, the group integral over the orthogonal group $\O(n)$ or unitary symplectic group give two other HCIZ integral (see e.g.~\cite{FEFZ07}), and the group integral over $\U(n)\times \U(n+\nu)$ is a Berezin-Karpelevich integral integral~\cite{BK58,GW96}. With the notions introduced in section~\ref{s2.1}, a general integral formula for all such cases can be stated as follows, distinguished by a parameter $\nu\in\mathbb{N}_0\cup\{\pm1/2\}$,
\begin{equation}
	\begin{split}
\int_{K_\nu} \mu(\dv K)\,&\exp\bigg(i\tr\, K\,\iota(x)\,K^{-1}\,\iota(s)\bigg)
\\
&=\left(\prod_{j=0}^{n-1}j!\Gamma(j+\nu+1)\right)\frac{\det\left[J_\nu(2\sqrt{x_js_k})/(x_js_k)^{\nu/2}\right]_{j,k=1}^n}{\Delta(x)\Delta(s)}=:\phi(x,s).\label{2.2.6}
\end{split}
	\end{equation}
The function $J_\nu$ denotes the Bessel function of the first kind with parameter $\nu$. Although we use the same notation $\phi$ for different spherical functions, it should be clear from the context which matrix group and, hence, which spherical transform we talk about.

When we equip the matrix space $M_\nu$ with the group action of matrix addition, we can state its spherical transform as follows, which is a matrix analogue of the multivariate Hankel transform. 

\begin{definition}[Spherical transform on $M_\nu$]
	Let $X\in M_\nu$ be an invariant random matrix with distribution $F$ and joint probability density $f\in L^1(\R^n)$ of its squared singular values. The spherical transform for an invariant $M_\nu$ matrix is given by
	\begin{equation}\label{2.2.7}
	\s f(s):=\int_{M_\nu}\dv X\, F(X)\exp(i\,\tr X\iota(s))=\int_{\R_+^n}\dv x\,f(x)\phi(x,s),
	\end{equation}
	where $\phi$ is given in~\eqref{2.2.6}, and the notation $\iota$ is introduced in section~\ref{s2.1}.
\end{definition} 

Also, the convolution formula~\eqref{2.2.5a} with the spherical transform on $M_\nu$ still applies when choosing invariant random matrices $A,B\in M_\nu$.

Since an inverse transform of~\eqref{2.2.7} is not explicitly written up we will state it for the particular case we consider in the ensuing proposition and prove it in Appendix~\ref{aa}. Indeed the formula can be easily derived from the results in~\cite{FKK17}. We are aware that it can be certainly extended to more general functions and even distributions.

\begin{proposition}[Inverse spherical transform for $M_\nu$]\label{p2.2.1}
	Let $F\in L^1(M_\nu)$ be an invariant function under the conjugation action of the group $K_\nu$ and $f\in L^1(\R_+^n)$ is related to $F$ via~\eqref{1.1.10}. With $\phi$ given in~\eqref{2.2.6}, the inverse spherical transform for the Hankel class is given by
	\begin{equation}\label{2.2.10}
	\s^{-1}[\s f](x)=\lim_{\varepsilon \to0^+}\frac{\Delta(x)^2}{(n!C_\nu)^2}\int_{\R_+^n}\dv s\,\s f(s)\phi(x,s)\Delta(s)^2\prod_{j=1}^n(x_js_j)^{\nu}e^{-\varepsilon s_j},
	\end{equation}
	with $C_\nu$ as in~\eqref{1.1.10}.
	The regularisation $\prod_{j=1}^ne^{-\varepsilon s_j}$ in $\varepsilon >0$ is anew only introduced for the cases when the remaining integrand is not absolutely integrable.
\end{proposition}

Let us remark that for $n=1$, the equations~\eqref{2.2.7} and~\eqref{2.2.10} reduce to the $n=1$ version of~\eqref{2.2.8}, which are the univariate Hankel transform and its inverse.

In contrast to the two additive cases, the multiplicative convolution has to be related to the Mellin transform. Its matrix version is given by the spherical transform on $\herm_+(n)$. In order to define this specific transform, we introduce with $|X|^s$ the generalised power function
\begin{equation}\label{2.3.1}
|X|^s=\prod_{j=1}^{n-1}\det X_{j\times j}^{s_j-s_{j+1}-1}\det X^{s_n},
\end{equation}
where $X_{j\times j}$ denotes the top left $j\times j$ block of $X\in\herm_+(n)$.
The parameters $s_n$ can be chosen complex. This generalised power function is essentially used in the multiplicative analogue of the HCIZ integral which is the Gelfand-Na\u\i mark integral (see~\cite{GN57}),
\begin{equation}\label{2.3.3}
\int_{\U(n)}\mu(\dv K)\,|K\diag(x)K^{-1}|^s=\left(\prod_{j=0}^{n-1}j!\right)\frac{\det[x_j^{s_k}]_{j,k=1}^n}{\Delta(x)\Delta(s)}=:\phi(x,s).
\end{equation}
When the eigenvalues $x$ or the parameters $s$ degenerate we need to apply l'H\^opital's rule.

Once the spherical function has been identified one can introduce the spherical transform on $\herm_+(n)$ which will reduce to the univariate Mellin transform for $n=1$. 

\begin{definition}[Spherical transform on $\herm_+(n)$]
	For $X\in\herm_+(n)$ unitarily invariant with distribution $F$ and joint probability density $f$ of the eigenvalues $x\in\mathbb{R}_+^n$, its spherical transform is defined as
	\begin{equation}\label{2.3.2}
	\s f(s):=\int_{\herm_+(n)}\frac{\dv X}{|\det X|^n}F(X)|X|^s=\int_{\R_+^n}\frac{\dv x}{\prod_{j=1}^nx_j^n}f(x)\phi(x,s).
	\end{equation}
	with $\phi$ given in~\eqref{2.3.3}. Let $s_0=(0,\cdots,n-1)$ be a fixed vector. The inverse spherical transform is written as
	\begin{equation}\label{2.3.5}
	\begin{split}
	\s^{-1}[\s f](x)&=\lim_{\varepsilon\to0}\frac{(-1)^{n(n-1)/2}}{\prod_{j=0}^{n}(j!)^2}\Delta(x)^2\\
	&\quad\times\int_{\R^n}\frac{\dv s}{(2\pi)^n}\Delta(s_0+is)^2\phi(x^{-1},s_0+is)\s f(s_0+is)\prod_{j=1}^ne^{\varepsilon(s_{0,j}+is_j)^2}.
	\end{split}
	\end{equation}
\end{definition}

The measure $dX/|\det X|^n$ is the Haar measure that corresponds to the multiplicative action $X\mapsto A^{1/2}XA^{1/2}$ for an arbitrary $A\in\herm_+(n)$. Thus, it is very natural to appear here, as the analogue of $\prod_{j=1}^n\dv x_j/x_j$ in the multivariate Mellin transform.

	The corresponding multiplicative convolution theorem to this transform takes a bit to get used to since the product $A^{1/2}BA^{1/2}$ of the two invariant random matrices $A, B\in\herm_+(n)$ with joint eigenvalue distributions $f_A$ and $f_B$ does not form a group, as we have pointed out earlier. Nonetheless, we have
	\begin{equation}
	\s f_{A^{1/2}BA^{1/2}}=\s f_A\cdot \s f_{B}
	\end{equation}
	which means that on the level of joint probability distributions of their eigenvalues this construction can be extended to a semi-group with the operation being the multiplicative convolution, and the Dirac delta distribution at the identity matrix being the unit element.

The spherical transform for our last case $\U(n)$ is very similar to the multiplication on $\herm_+(n)$. To construct it, we need to introduce the group characters of the irreducible representations of $\U(n)$ which are the Schur polynomials. It has the explicit form $\phi(e^{i\theta},s)$ defined in~\eqref{2.3.3}, where $\theta=(\theta_1,\ldots,\theta_n)\in(-\pi,\pi]^n$ and $s=(s_1,\ldots,s_n)\in\mathbb{Z}^n$, where the components of $s$ need to be pairwise distinct. As can be readily seen, for $n=1$, we obtain the well-known Fourier factor.

\begin{definition}[Spherical transform on $\U(n)$]
	Let $e^{i\theta}$ be the eigenvalues of $X\in\U(n)$ distributed along $f\in L^1((-\pi,\pi]^n)$, which is $2\pi$ periodic in each entry. Its spherical transform reads
	\begin{equation}\label{2.3.7}
	\s f(s)=\int_{(-\pi,\pi]^n}\dv \theta\,f(e^{i\theta})\phi(e^{i\theta},s),
	\end{equation}
	while its inverse is
	\begin{equation}\label{2.3.8}
	\s^{-1}(\s f)(e^{i\theta})=\lim_{\varepsilon\to0}\frac{1}{(2\pi)^n\prod_{j=0}^nj!^2}|\Delta(e^{i\theta})|^2\sum_{s\in\Z^n}\s f(s)\phi(e^{-i\theta},s)\Delta(s)^2\prod_{j=1}^ne^{-\varepsilon s_j^2},
	\end{equation}
	see e.g.~\cite{ZKF19}. The sum over $s$ can be restricted to those where the $s_j$'s are pairwise different since the Vandermonde determinant $\Delta(s)$ vanishes then.
\end{definition}

There are several consequences from the fact that $\phi(e^{-i\theta},s)$ is essentially a group character. One of these is the convolution theorem which reads for two invariant random matrices $A,B\in\U(n)$ with the joint probability distributions $f_A$ and $f_B$ of their respective eigenvalues as follows
	\begin{equation}
	\s f_{AB}=\s f_A\cdot \s f_{B}.
	\end{equation}
The joint probability distribution $f_{AB}$ is the one of the eigenvalues of the product $AB$.

A useful expression of the spherical transform in terms of a matrix integral can be derived when combining~\eqref{2.3.7} with~\eqref{2.3.1} and~\eqref{2.3.3} when assuming $s_1>s_2>\ldots>s_n$. This expression reads
\begin{equation}\label{2.3.7a}
\s f(s)=\int_{\U(n)}\mu({\dv X})\,F(X)|X|^s.
\end{equation}
We will make use of it when deriving the derivative \textit{principle} for invariant random matrix ensembles on the unitary group. The restriction of $s$ is indispensable as otherwise the integral~\eqref{2.3.7a} may run through poles given by the principal minors of $X$. In contrast to $\herm_+(n)$, the principal minors can vanish while the one of a positive definite matrix do not. Hence, one may ask whether there is a problem for the spherical transform and, especially its inverse on $\U(n)$ as we have to sum over all $s\in\mathbb{Z}^n$ with pairwise different components. This can be resolved, as one can extend~\eqref{2.3.7a} to all arrangements of $s$ by the symmetry relation $\s f(s)=\s f(s_\rho)$ for any permutation $\rho\in\mathrm{S}_n$, obtained from the definition~\eqref{2.3.7}.

Now, we are ready with the preparations and go over to derive the derivative principles.

\section{Additive invariant ensembles}\label{s3}

\subsection{Derivative principle on $\herm(n)$}\label{s3.1}

To display the idea behind the derivative principle we will briefly review the case of invariant Hermitian matrices $\herm(n)$. It essentially relates the joint probability distribution $f$ of the eigenvalues to the joint distribution $f_{\rm diag}$ of the diagonal entries of $X$. Proofs of this relation can be found in~\cite{CDKW14,MZB16,Fa06}, which we will also outline here to compare it with the derivative principles in the other matrix spaces.

Let $X\in\herm(n)$ be a random matrix drawn from the invariant probability distribution $F$. Its joint probability density $f$ of its eigenvalues $x=(x_{1},\ldots,x_{n})$ is given by~\eqref{1.1.2}. The marginal distribution of the diagonal entries $\tilde{x}=(x_{11},\ldots,x_{nn})$ of $X=\{x_{j,k}\}_{j,k=1,\ldots,n}$ is equal to
\begin{equation}\label{3.1.1}
f_\mathrm{diag}(\tilde{x}):=\int F(X)\prod_{1\le j<k\le n}\dv x_{j,k}.
\end{equation}
In preparation for further analysis, we require that $f_\diag$ has to be in $L^{1,n(n-1)/2}_{\f}(\R^n)$, especially it has to be suitably differentiable. We are certain that this condition can be slightly relaxed as the P\'olya ensembles discussed in~\cite[eqn. (11)]{Ki19} can be traced back to a univariate density which has to be only $(n-1)$-times differentiable. We now restate Propsition~\ref{p1.1} as follows.

\begin{proposition}[Derivative principle for $\herm(n)$~\cite{CDKW14,MZB16,Fa06}]\label{p3.1.1}
	Let $X\in\herm(n)$ be a $\U(n)$-invariant matrix with a joint probability distribution $f\in L^1(\R^n)$ of the eigenvalues $x$ and a distribution $f_\diag\in L^{1,n(n-1)/2}_{\f}(\R^n)$ of its diagonal entries. Then, the two distributions satisfy the relation
	\begin{equation}\label{3.1.5}
	f(x)=\frac{1}{\prod_{j=0}^nj!}\Delta(x)\Delta(-\partial_{x})f_\mathrm{diag}(x)
	\end{equation}
	for almost all $x\in\mathbb{R}^n$.
	Here, we understand the Vandermonde determinant $\Delta(-\partial_x):=\prod_{1\le j<k\le n}(\partial_{x_j}-\partial_{x_k})$ as a polynomial in the partial derivatives.
\end{proposition} 

\begin{proof}
	As the Fourier factor in the integral of the spherical transform of $f$ does not depend on the off-diagonal entries of $X$, see~\eqref{2.2.4}, we first integrate over all these off-diagonal entries. Again with $\tilde{x}=(x_{1,1},\ldots,x_{n,n})$ this yields
	\begin{equation}\label{3.1.6}
	\s f(s)=\int_{\R^n}\dv \tilde{x}\,\exp\left(i\sum_{j=1}^nx_{j,j}s_j\right)f_\mathrm{diag}(\tilde{x})=\f f_\diag(s).
	\end{equation}
	Next, we substitute~\eqref{3.1.6} into the inverse transform~\eqref{2.2.5} which explicitly reads
	\begin{equation}\label{3.1.7}
	f(x)=\frac{\Delta(x)}{n!\prod_{j=0}^nj!}\int_{\R^n}\frac{\dv s}{(2\pi)^n} \,\Delta(is)\det[e^{-ix_js_k}]_{j,k=1}^{n}\f f_\diag(s).
	\end{equation}
	The Laplace expansion in the second determinant of the integrand yields
	\begin{equation}
	f(x)=\frac{\Delta(x)}{\prod_{j=0}^nj!}\,\frac{1}{n!}\sum_{\rho\in \mathrm{S}_n}\int_{\R^n}\frac{\dv s}{(2\pi)^n}\,\Delta(is)\,\sgn(\rho)\exp\left(-i\sum_{j=1}^nx_{\rho(j)}s_j\right)\f f_\diag(s) ,
	\end{equation}
	where $\mathrm{S}_n$ denotes the symmetric group of order $n$ and $\mathrm{sgn}(\rho )$ is the signum function which is $1$ for even permutations and $-1$ for odd permutations. One can replace $\Delta(is)\sgn(\rho)$ by $\Delta(-\partial_{x})$ as its action on the Fourier terms yields the proper polynomial in $s$. The interchange of the derivatives with the $s$-integral is allowed, because $\f f_\diag$ is in $L^1(\mathbb{R}^n)$ resulting from the differentiability of $f_\diag$ and the phases do not change this fact. The remaining $s$-integral is an inverse multivariate Fourier transform of $\f f_\diag$, so that we eventually have
	\begin{equation}
	f(x)=\frac{1}{\prod_{j=0}^nj!}\Delta(x)\Delta(-\partial_{x})\frac{1}{n!}\sum_{\rho\in \mathrm{S}_n}f_\mathrm{diag}(x_\rho)
	\end{equation}
	for almost all $x\in\mathbb{R}^n$. The invariance of $F$ bijectively relates $F$ to $f$ and also implies a permutation symmetry of the arguments of $f_{\rm diag}$, i.e., $f_\mathrm{diag}(x_\rho)=f_\mathrm{diag}(x)$, so that the sum over $\rho$ reduces to a factor $n!$,
	which completes the proof.
\end{proof}

As illustrated in~\cite{MZB16}, Proposition~\ref{p3.1.1} allows us to compute sums of random matrices in a natural way. Let us consider three invariant random matrices $A,B,C\in\herm(n)$, satisfying $A+B=C$ with the joint probability distributions $f_\diag^{(A)}, f_\diag^{(B)}, f_\diag^{(C)}$ of their diagonal entries, respectively. It is trivial to see that diagonal entries also add component-wise up in this sum, so that they satisfy the multivariate convolution relation
\begin{equation}\label{3.1.9}
f_\diag^{(A)}\ast f_\diag^{(B)}=f_\diag^{(C)}.
\end{equation}
This relation reduces a sum of two independent random matrices into a sum of two independent random vectors (the diagonal entries). We can combine~\eqref{3.1.9} with the derivative principle~\eqref{3.1.5} to obtain the joint probability distribution $f_C$ of the eigenvalues of $C$. It becomes particularly transparent, when considering that the multivariate Fourier transform of $f_\diag^{(C)}$ is the same as the spherical transform of $f_C$
\begin{equation}
\f f_\diag^{(C)}=\s f_C=\s f_A\cdot \s f_B=\f f_\diag^{(A)}\cdot \f f_\diag^{(B)}=\f (f_\diag^{(A)}\ast f_\diag^{(B)}).
\end{equation}
Taking $\f^{-1}$ on both sides reclaims~\eqref{3.1.9}.

\begin{corollary}[Additive convolution on $\herm(n)$]\label{corr.addHerm}
Let $A,B\in\herm(n)$ be two independent invariant random matrices whose joint distributions $f_\diag^{(A)}, f_\diag^{(B)}$ of their diagonal elements satisfy the requirements of Proposition~\ref{p3.1.1}. Then, the joint probability density $f_C$ of the eigenvalues $x$ of $C=A+B$ is given by
\begin{equation}
f_C(x)=\frac{1}{\prod_{j=0}^nj!}\Delta(x)\Delta(-\partial_{x})[f_\diag^{(A)}\ast f_\diag^{(B)}](x).
\end{equation}
\end{corollary}

A natural question is whether there exists another function than $f_\diag$, which captures the additive nature of random matrix addition. We show in the following lemma that such a function is unique, if one requires it to be in $L^{1,n(n-1)/2}(\R^n)$.

\begin{lemma}\label{p3.1.2}
	Let $u\in L^{1,n(n-1)/2}_\f(\R^n)$, and $P(\partial_{x})\not\equiv 0$ denotes a polynomial of partial derivatives $\partial_{x_1},\ldots,\partial_{x_n}$. Then the following partial differential equation
	\begin{equation}
	P(\partial_{x})u(x)=0\quad {\rm for\ all}\ x\in\mathbb{R}^n
	\end{equation}
	gives the unique solution $u(x)=0$.
\end{lemma}

\begin{proof}
	Taking the Fourier transform on both sides gives $P(s)\f u(s)=0$, which implies $\f u(s)=0$ for all $s\in\mathbb{R}^n$ with $P(s)\ne 0$. As $P$ is a polynomial, the points where $P(s)=0$ build a set of measure zero. By continuity one can extend this to all $s\in\mathbb{R}^n$. The injectivity of the Fourier transform on $L^{1}(\R^n)$ yields the claim. 
\end{proof}

Combining Propositions~\ref{p3.1.1} and Lemma~\ref{p3.1.2}, one can conclude the following uniqueness of the derivative principle for additive invariant ensembles on $\herm(n)$. Hence, there is a one-to-one correspondence between $f$ and $f_{\rm diag}$ which is remarkable as in the latter we integrate over all off-diagonal entries of the random matrix whose distributions do not necessarily factorise from the other ones.

\begin{corollary}[Uniqueness of additive invariant ensemble on $\herm(n)$]\label{p3.1.3}
	Considering the setting of Proposition~\ref{p3.1.1}. Then, there exists a unique symmetric function $w\in L^{1,n(n-1)/2}_{\f}(\R^n)$ such that
	\begin{equation}\label{3.1.12}
	f(x)=\frac{1}{\prod_{j=0}^nj!}\Delta(x)\Delta(-\partial_x)w(x).
	\end{equation}
	We refer to $w$ as the additive weight. Additionally, one has the relation
	\begin{equation}\label{3.1.12a}
	\s f=\f w,
	\end{equation}
	where $\s f$ denotes the spherical transform of $f$, and $\f w$ denotes the multivariate Fourier transform of $w$, and $w=f_\diag$.
\end{corollary}

We will give a few examples for invariant ensembles on $\herm(n)$. The first examples are the additive P\'olya ensembles introduced in~\cite{KR19,FKK17}, which cover a broad class of classical matrix ensembles. The second class of examples do not necessarily belong to this set of ensembles but are invariant ensembles, nonetheless.

\begin{example}[Examples for invariant ensembles on $\herm(n)$]\label{p3.1.4}$\,$
	\begin{enumerate}
		\item (Additive P\'olya ensemble on $\herm(n)$~\cite{KR19,FKK17}) An additive P\'olya ensemble on $\herm(n)$ is a subclass of invariant ensemble, whose diagonal entries are independent. By the invariance under $\U(n)$-action, they are also identically distributed. The weight function therefore has the following expression $w(x)=\prod_{j=1}^n\tilde w(x_j)$, where $\tilde w$ is the distribution for each diagonal entry $x_j$. This class contains many commonly seen random matrix ensembles including Gaussian unitary ensemble, Wishart-Laguerre ensemble, some Muttalib-Borodin ensembles etc. Let us underline that it is striking that the joint probability distribution of the diagonal entries of a P\'olya ensemble is factorising so that the diagonal entries are identically and independently distributed random variables.
		
		\item (Polynomial ensemble) Polynomial ensembles introduced in~\cite{KR19} have a joint probability distribution of the eigenvalues which are proportional to
		\begin{equation}
		f(x)\propto \Delta(x)\det\left[w_j(x_k) \right]_{j,k=1}^n,
		\end{equation}
		for some appropriate weight functions $w_1, w_2, \ldots, w_n$. To formalise this as an additive ensemble one needs to solve the following PDE
		\begin{equation}
		\Delta(-\partial_x)w(x)=\det\left[w_k(x_j) \right]_{j,k=1}^n,
		\end{equation}
		for $w\in L^{1,n(n-1)/2}_{\f}(\R^n)$. This can be solved by applying the Fourier transform on both sides, dividing by the Vandermonde determinant, and applying the inverse transform, i.e.
		\begin{equation}\label{polya-add}
		w(x)=\f^{-1}\left[\frac{\det\left[\f w_k(s_j) \right]_{j,k=1}^n}{\Delta(s)}\right](x).
		\end{equation}
		Unfortunately, this integral cannot be resolved in full generality and it only reduces to simple expressions for particular classes like the P\'olya ensembles on $\herm(n)$. For instance, the Jacobi ensemble, where $w_k(x_j)=x_j^{k-1+a}(1-x_j)^b H(x)H(1-x)$ for $a,b>0$ and $H$ is the Heaviside step function, does not belong to a P\'olya ensemble on $\herm(n)$ equipped with the matrix addition, although it is a P\'olya ensemble on $\herm_+(n)$ when it is equipped with the symmetric matrix multiplication $(A,B)\mapsto A^{1/2}BA^{1/2}$, see~\cite{KK19} and the second case in Example~\ref{p4.1.4}.
	\end{enumerate}
\end{example}

\subsection{Derivative principles on $i\o(2n)$, $i\o(2n+1)$ and $i\usp(2n)$}\label{s3.2}

Derivative principles on $i\o(2n)$, $i\o(2n+1)$ and $i\usp(2n)$ are analogous to the one on $\herm(n)$. We will first give the derivative principles for $i\o(m)$ ensembles in Proposition~\ref{p3.2.1} with $m=2n$ being even as well as $m=2n+1$ being odd. The case $M_{1/2}=i\usp(2n)$ is very similar to case $M_{1/2}=i\o(2n+1)$ as their very similar Dynkin diagrams imply the same eigenvalue statistics apart from the generic zero eigenvalues (the length of the roots cancel out after normalising of the distribution). It will be discussed in Corollary~\ref{p3.2.4}.

As we have seen in the previous section, we could relate the joint probability distribution of the eigenvalues with the one of the diagonal entries. In the present case, the eigenvalues will be replaced by the squared singular values of an invariant random matrix $X\in M_{\pm1/2}=i\o(m)$ that is drawn from $F\in L^{1}(i\o(m))$. The question is what will be the diagonal entries as $X=-X^T$ has no non-zero entries on its diagonal. What we need are the pseudo-diagonal entries $y=(-ix_{1,2},-ix_{3,4},\ldots,-ix_{2n-1,2n})$ for $l=1,\ldots,n$ which are essentially the image of the embedding $\iota: \mathbb{R}_+^n\to i\o(m)$; indeed this view needs to be extended to a vector space.
Despite its subscript, let $f_\diag$ be the marginal distribution for these pseudo-diagonal entries, especially it is given by
\begin{equation}\label{3.2.1}
f_\diag(y)=\int F(X)\prod_{\substack{1\le j<k\le m\\(j,k)\ne(2l-1,2l)}}\dv x_{j,k},\quad (l=1,\ldots,n).
\end{equation}
As introduced in section~\ref{s2.1}, the $i$ factor assures those variables are real because $X\in M_{\pm1/2}=i\o(m)$ has purely imaginary entries.

The invariance of $F$ under the orthogonal group ${\rm O}(m)$ implies that the function $f_\diag$ is not only permutation invariant but also even in each of its variables. Furthermore, we require the regularity condition $f_\diag\in L^{1,n(n-1)}_{\f}(\R^n)$ defined in~\eqref{2.2.4a} as we want to apply again a derivative operator on the function. We then have the following two derivative principles for even and odd $m$.

\begin{proposition}[Derivative principle for $i\o(m)$]\label{p3.2.1}
	Let $X\in i\o(m)$ be an $\O(m)$-invariant random matrix whose squared singular values $x\in\mathbb{R}_+^n$ are drawn from $f\in L^1(\R_+^n)$ and its pseudo-diagonal entries $\tilde{x}$ follow the joint distribution $f_\diag\in L^{1,n^2+n(\nu-1/2)}_{\f}(\R^n)$. Then, those two functions satisfy the relations
	\begin{equation}\label{3.2.3}
	f(x)=\frac{\pi^{n/2}}{n!C_{-1/2}}\Delta(x)\Delta(-x^{-1/2}\partial_{x}x^{3/2}\partial_x)\frac{f_\diag(\sqrt{x_1},\ldots,\sqrt{x_n})}{\prod_{j=1}^n\sqrt{x_j}}
	\end{equation}
	for $m=2n$ and
	\begin{equation}\label{3.2.3a}
	f(x)=\frac{\pi^{n/2}}{n!C_{1/2}}\Delta(x)\Delta(-x^{1/2}\partial_{x}x^{1/2}\partial_x)\left(\prod_{j=1}^n-x_j^{1/2}\partial_{x_j}\right)\,f_\diag(\sqrt{x_1},\ldots,\sqrt{x_n})
	\end{equation}
	for $m=2n+1$. Both relations hold only 
	for almost all $x\in\mathbb{R}_+^n$, and the constants $C_{\pm 1/2}$ are given by~\eqref{1.1.10} in the $i\o(m)$ cases.
\end{proposition} 

\begin{proof}
	By integrating over all but the pseudo-diagonal entries $y=(-ix_{1,2},-ix_{3,4},\ldots,-ix_{2n,2n-1})$, the spherical transform of $f$ in~\eqref{2.2.7} can be written as
	\begin{equation}
	\begin{split}
	\s f(s)&=\int_{\R^n}\dv y\, f_\diag(y)\exp\left(2i\sum_{j=1}^ny_{j}\sqrt{s_j}\right)
	\\&=\f f_\diag(2\sqrt{s_1},\ldots,2\sqrt{s_n})=\int_{\R_+^n}\dv y\, f_\diag(y)\prod_{j=1}^n2\cos\left(2y_{j}\sqrt{s_j}\right).\label{3.2.4}
	\end{split}
	\end{equation}
	With $\f$ denoting the multivariate Fourier transform~\eqref{Fourier}, we also emphasise by the third equal sign that the Fourier transform of the even function $f_\diag$ coincides with its Fourier cosine transform. Substituting~\eqref{3.2.4} into the inverse spherical transform~\eqref{2.2.10} yields
	\begin{equation}
	\begin{split}\label{3.2.5}
	f(x)=\frac{\Delta(x)}{(n!)^2C_{\pm1/2}}\int_{\R_+^n}{\dv s} \,\Delta(s)\det[J_{\pm1/2}(2\sqrt{x_js_k})(x_js_k)^{\pm1/4}]_{j,k=1}^{n}\f f_\diag(2\sqrt{s_1},\ldots,2\sqrt{s_n}).
	\end{split}
	\end{equation}
	The regularisation could be dropped since the differentiability guarantees the absolute integrability of the integrand above.
	
	The determinant in the Bessel functions can be Laplace-expanded as a sum over permutations in the symmetric group of $n$ elements, $\mathrm{S}_n$. Thereafter, we can exploit that $J_{\pm1/2}(2\sqrt{x_js_{\rho(j)}})(x_js_{\rho(j)})^{\pm1/4}$ is an eigenfunction for the differential operator $-x_j^{\pm1/2}\partial_{x_j}x_j^{1\mp 1/2}\partial_{x_j}$, with the eigenvalue $s_{\rho(j)}$. Therefore, $\Delta(s)\sgn(\rho)$ can be replaced by $\Delta(-x^{\pm1/2}\partial_{x}x^{1\mp1/2}\partial_x)$. Additionally, we change variables $2\sqrt{s_j}\mapsto t_j$ so that we get
	\begin{equation}
	\begin{split}
	f(x)=&\frac{\Delta(x)}{2^{n(1\pm1/2)}n!C_{\pm1/2}}\,\frac{1}{n!}\sum_{\rho\in \mathrm{S}_n}\int_{\R_+^n}\dv t_j\,\Delta(-x^{\pm1/2}\partial_{x}x^{1\mp1/2}\partial_x)
	\\&\times 
	\left(\prod_{j=1}^nJ_{\pm1/2}(\sqrt{x_{\rho(j)}}t_{j}) {x_{\rho(j)}}^{\pm1/4}t_{j}^{1\pm1/2}\right)\f f_\diag(t_1,\ldots,t_n).\label{3.2.8}
	\end{split}
	\end{equation}
	
	Eventually, we make use of the explicit representation of the Bessel function in terms of trigonometric functions. For the case $\nu=-1/2$ it is \begin{equation}
	J_{-1/2}(\sqrt{x_{\rho(j)}}\,t_{j}){x_{\rho(j)}}^{-1/4}t_{j}^{1/2}=\sqrt{2/(\pi x_{\rho(j)})}\cos(\sqrt{x_{\rho(j)}}t_{j}).
	\end{equation} 
	Then the $y$-integral gives the multivariate Fourier cosine transform of $f_\diag$,
	\begin{equation}
	\begin{split}\label{3.2.9}
	f(x)=\frac{\Delta(x)}{n!C_{-1/2}}\,\frac{1}{n!}\sum_{\rho\in \mathrm{S}_n}\int_{\R_+^n}\dv t\,\Delta(-x^{-1/2}\partial_{x}x^{3/2}\partial_x)
	\left(\prod_{j=1}^n\frac{\cos(\sqrt{x_{\rho(j)}}t_{j})}{\sqrt{\pi x_{\rho(j)}}}\right)\f f_\diag(t_1,\ldots,t_n).
	\end{split}
	\end{equation}
	Again we exploit the absolute integrability of $\f f_\diag$ so that we can interchange the differential operator $\Delta(-x^{-1/2}\partial_{x}x^{3/2}\partial_x)$ with the $t$-integral. This remaining integral is an inverse multivariate Fourier cosine transform which leads to
	\begin{equation}
	f(x)=\frac{\pi^{n/2}}{n!C_{-1/2}}\Delta(x)\Delta(-x^{-1/2}\partial_{x}x^{3/2}\partial_x)\frac{1}{n!}\sum_{\rho\in \mathrm{S}_n}\frac{f_\diag(\sqrt{x_{\rho(1)}},\ldots,\sqrt{x_{\rho(n)}})}{\prod_{j=1}^n\sqrt{x_{\rho(j)}}}.
	\end{equation}
	By the permutation invariance of $f_\diag$ one reclaims~\eqref{3.2.3}.
	
 The same computation can be carried out for the case $\nu=1/2$, where we have now \begin{equation}
 J_{1/2}(\sqrt{x_{\rho(j)}}t_{j})x_{\rho(j)}^{1/4}t_{j}^{3/2}=\sqrt{2/\pi}\,t_j\sin(\sqrt{x_{\rho(j)}}t_{j}).
 \end{equation} 
 The sine function can be rewriting in terms of the first derivative of a cosine so that Eq.~\eqref{3.2.9} becomes
	\begin{equation}
	\begin{split}
	f(x)=&\frac{\Delta(x)}{2^{3n/2}n!C_{1/2}}\,\frac{1}{n!}\sum_{\rho\in \mathrm{S}_n}\int_{\R_+^n}\dv t\Delta(-x^{1/2}\partial_{x}x^{1/2}\partial_x)
	\\
	&\times\left(\prod_{j=1}^n-\frac{2^{3/2}}{\sqrt{\pi}}\sqrt{x_j}\partial_{x_j}\cos(\sqrt{x_j}t_{\rho(j)})\right)\f f_\diag(t_1,\ldots,t_n).
	\end{split}
	\end{equation}
	After switching all derivatives with the $t$-integral, one obtains an inverse multivariate Fourier cosine transform, and after applying the invariance of $f_\diag$, one obtains
	\begin{equation}
	\begin{split}
	f(x)&=\frac{\pi^{n/2}}{n!C_{1/2}}\Delta(x)\Delta(-x^{1/2}\partial_{x}x^{1/2}\partial_x)\left(\prod_{j=1}^n-\sqrt{x_j}\partial_{x_j}\right)\\
	&\quad\times\frac{1}{n!}\sum_{\rho\in \mathrm{S}_n}f_\diag(\sqrt{x_{\rho(1)}},\ldots,\sqrt{x_{\rho(n)}}).
	\end{split}
	\end{equation}
	By invariance of $f_\diag$,~\eqref{3.2.3a} is reclaimed.
\end{proof}

The normalisation of the two expression on the right hand side of~\eqref{3.2.3} and~\eqref{3.2.3a} can be checked by expanding the two determinants and integration by parts. What remains is in both cases the integral
\begin{equation}
\int_{\mathbb{R}_+^n}\frac{dx_1}{\sqrt{x_1}}\cdots\frac{dx_n}{\sqrt{x_n}}f_\diag(\sqrt{x_1},\ldots,\sqrt{x_n})=1.
\end{equation}
Indeed, $f_\diag(\sqrt{x_1},\ldots,\sqrt{x_n})/\sqrt{\prod_{j=1}^nx_j}$ is the distribution of the squared pseudo-diagonal entries $(-ix_{1,2},-ix_{3,4},\ldots,-ix_{2n,2n-1})$. In the particular case of $m=2$, we obtain a trivial equation since the matrix space is one-dimensional and we average only over the sign of $-ix_{1,2}$ which drops out in any case as $f_\diag$ is an even function. For the next simple case $m=2n+1=3$ we obtain from~~\eqref{3.2.3a} the remarkable identity
	\begin{equation}
	f(x_1)=-f_\diag'(\sqrt{x_1})
	\end{equation}
	because $C_{1/2}|_{n=1}=\sqrt{\pi}/2$.
	The left hand side is non-negative so that the distribution of $x_{1,2}$ can only be monotonously decreasing on the positive real line.
	
	The case $m=2n+1=3$ relates the single squared singular value distribution to one of the marginal distributions of the matrix entries of an invariant imaginary antisymmetric matrix $X\in i\o(3)$. Plugging this knowledge into the general case $m=2n+1$, we notice that $f(x)$ essentially depends on a joint probability distribution of eigenvalues of a list of in general correlated invariant $i\o(3)$ matrices $X_1,\ldots,X_n$, and their pseudo-diagonal entries $x_{1,2}^{(1)},\ldots,x_{1,2}^{(n)}$ of those matrices have the joint distribution $f_\diag$.
	
	Furthermore, we would like to point out that $i\o(m)$ matrices can be regarded as $M_\nu$ matrices with $\nu=\pm1/2$. Then, both~\eqref{3.2.3} and~\eqref{3.2.3a} can be rewritten into a unified expression
	\begin{equation}\label{3.2.13a}
	f(x)=\frac{1}{n!C_\nu}\Delta(x)\Delta(-x^\nu\partial_xx^{1-\nu}\partial_x){\mathcal A}_\nu^{-1} f_\diag(x),
	\end{equation}
	where ${\mathcal A}_\nu$ is the multivariate inverse Abel transform~\eqref{2.2.8a}. This formula can be implicitly seen in the proof of Proposition~\ref{p3.2.1}. For example for $\nu=-1/2$, Eq.~\eqref{3.2.9} can be rewritten by pulling out the differential operator. Then, the $t$-integral is an inverse Hankel transform. Similar strategies applies to $\nu=1/2$ case. We point out formula~\eqref{3.2.13a} since it caries over to $M_\nu=\Mat(n,n+\nu)$ which will be discussed in the next section.

\begin{remark}\label{p3.2.4a}
	We would like to add one final remark before we briefly restate Proposition~\ref{p3.2.1} for the case $M_{1/2}=i\usp(2n)$. Equations~\eqref{3.2.3} and~\eqref{3.2.3a} can be rewritten into the joint distribution $f_{\mathrm{ev}}$ of the non-zero eigenvalues $\pm\lambda$ of the antisymmetric matrix $X\in i\o(m)$ with $\lambda=(\lambda_1,\ldots,\lambda_n)\in\mathbb{R}^n$ and $m=2n$ or $m=2n+1$. This distribution is either 
	\begin{align}
	\label{3.2.13}
	f_{\mathrm{ev}}(\lambda)&=
	\displaystyle\frac{\pi^{n/2}}{2^{n(n-1)} n!C_{-1/2}}\Delta(\lambda^2)\Delta(-\partial_\lambda^2)f_\diag(\lambda)
	\end{align}
	for even $m=2n$ or 
	\begin{align}
	f_{\mathrm{ev}}(\lambda)&=\frac{\pi^{n/2}}{2^{n^2}n!C_{1/2}}\left(\prod_{j=1}^n\lambda_j\right)\Delta(\lambda^2)\Delta(-\partial_\lambda^2)\left(\prod_{j=1}^n-\partial_{\lambda_j}\right)f_\diag(\lambda) \label{3.2.14}
	\end{align}
	for odd $m=2n+1$.
	Here, the square functions for both $\lambda^2$ and $\partial_\lambda^2$ are taken entry-wise. Indeed after changing variables $x\mapsto \lambda^2$ and noticing that $\lambda$ can be positive as well as negative while $x$ is always non-negative we obtain
	\begin{equation}
	f_\mathrm{ev}(\lambda)=f(\lambda^2)\prod_{j=1}^n|\lambda_j|.
	\end{equation}
	Moreover, one needs
	\begin{equation}
	\begin{split}
	x_j^{-1/2}\partial_{x_j} x_j^{3/2}\partial_{x_j}=\frac{1}{4}x_j^{-1/2}\partial_{\sqrt{x_j}}^2\sqrt{x_j}\quad{\rm and}\quad x_j^{-1/2}\partial_{x_j} x_j^{1/2}\partial_{x_j}=\frac{1}{4}\partial_{\sqrt{x_j}}^2.
	\end{split}
	\end{equation}
	as well as $x^{1/2}_j\partial_{x_j}=2^{-1}\partial_{\sqrt{x_j}}$.
\end{remark}

Despite the difference between $i\usp(2n)$ and $i\o(2n+1)$ matrix models on the level of matrices, those two cases are almost identical in terms of their squared singular value distributions and spherical transforms. Therefore without further discussion, one can find the analogous result for $i\usp(2n)$.

\begin{corollary}[Derivative principle for $i\usp(2n)$]\label{p3.2.4}
	Let the random matrix $X\in i\usp(2n)$ be invariant under $\USp(2n)$ conjugation action and $f_\diag\in L^{1,n(n-1)}_\f(\R^n)$ shall be the distribution of its diagonal entries $y=(x_{2,2},\ldots,x_{2n,2n})$, i.e.
	\begin{equation}\label{3.2.18}
	f_\diag(y_1,\ldots,y_n)=\int F(X)\prod_{1\le j\le k\le n}\dv q_{j,k},\quad (l=1,\ldots,n).
	\end{equation}
	The $2\times 2$ blocks $q_{j,k}=q_{k,j}^\dagger$ are given in the quaternion form~\eqref{1.1.8} and $x_{2l,2l}=\{q_{l,l}\}_{2,2}$. The joint probability distribution $f$ of the squared eigenvalues and, the joint distribution $f_\mathrm{ev}$ of the eigenvalues are related to $f_\diag$ by~\eqref{3.2.3a} and~\eqref{3.2.14}, respectively.
\end{corollary}

\subsection{Derivative principle on $\Mat(n,n+\nu)$}\label{s3.3}

Let us come to the matrix spaces $M_\nu=\Mat(n,n+\nu)$ with $\nu\in\mathbb{N}_0$. When considering a $\U(n)\times \U(n+\nu)$ invariant random matrix $X\in M_\nu$ we know what the joint probability density $f$ of its squared singular values $x$ is in terms of the matrix density $F\in L^1(M_\nu)$, namely, it is given by~\eqref{1.1.10}.

Anew, we have to choose, in the first step, the new pseudo-diagonal entries, which are $y=({\rm Re}(x_{1,1}),\ldots,{\rm Re}(x_{n,n}))$. This choice is again very natural since it is the image of $\iota:\mathbb{R}_+^n\to M_\nu$ extended to a real vector space. Then, their marginal distribution $f_\diag$ is given by
\begin{equation}\label{3.3.1}
f_\diag(y)=\int F(X)\prod_{{j=1}}^n\dv x_{j,j}^{(i)}\prod_{\substack{j,k=1,\ldots,n\\j\ne k}}\dv x_{j,k}^{(r)}\dv x_{j,k}^{(i)},
\end{equation}
where $x_{j,k}^{(r)}={\rm Re}(x_{j,k})$ and $x_{j,k}^{(i)}={\rm Im}(x_{j,k})$ represent the real and imaginary parts of $x_{j,k}$ respectively. To guarantee the derivatives we need to choose $f_\diag\in L^{1,n(n-1)+n(\nu+2)}_\f(\R^n)$ , this time. The reason that the additional order of $n(\nu+2)$ is to guarantee that the inverse Abel transform satisfies $\mathcal{A}_\nu^{-1} f_\diag\in L^{1,n(n-1)}_\f(\R^n)$. As mentioned before, we will find essentially the derivative principle~\eqref{3.2.13a}. This order is also reflected in encountering the Fourier transform of the $\nu$-th derivative of $f$ in each variable when deriving our result.

One important property of $f_\diag$ is, like in $i\o(m)$ cases, that it is an even function in each argument as well as permutation invariant. This can be applied in a similar way as in~\eqref{3.2.4}. Thence, the spherical transform is a real transform (in fact the spherical function is real). Hence the right side of~\eqref{3.2.4} must also be a real function.

One can straightforwardly carry over the proof of Proposition~\ref{p3.2.1} to the present case. In particular, we will end up with a similar formula as in~\eqref{3.2.8} which can be expressed as
\begin{equation}
	\begin{split}
f(x)
&=\frac{\Delta(x)\Delta(-x^\nu\partial_{x}x^{1-\nu}\partial_x)}{n!C_{\nu}}\,\frac{1}{n!}\sum_{\rho\in \mathrm{S}_n}\h_\nu^{-1}[\f f_\diag(2\sqrt{s})](\sqrt{x_\rho}),\label{3.3.3a}
\end{split}
	\end{equation}
where $\h_\nu^{-1}$ is the inverse multivariate Hankel transform and $\f$ is the multivariate Fourier transform which combines to the inverse multivariate Abel transform $\mathcal{A}_\nu^{-1} $. Their composition explicitly reads
\begin{equation}\label{3.3.3}
\begin{split}
\h_\nu^{-1} [\f f_\diag(2\sqrt{s})](\sqrt{x_\rho})=&\mathcal{A}_\nu^{-1} f_\diag(\sqrt{x_\rho})\\
=&\int_{\R_+^n}\dv s_j\,\left(\prod_{j=1}^nJ_\nu(2\sqrt{x_{\rho(j)}s_j}) (x_{\rho(j)}s_{j})^{\nu/2}\right)\f f_\diag(2\sqrt{s_1},\ldots,2\sqrt{s_n})\\
=&\prod_{j=1}^nx_j^{\nu/2}\int_{\sqrt{x_1}}^\infty\dv y_1\ldots\int_{\sqrt{x_n}}^\infty\dv y_n\left(\prod_{j=1}^n\frac{1}{\sqrt{y_j-\sqrt{x_j}}}\partial_{y_j}^{\nu+1}\right) f_\diag(\sqrt{y_1},\ldots,\sqrt{y_n}).
\end{split}
\end{equation} 
In the second line we have applied~\eqref{2.2.8a}.
The permutation $\rho$ drops out due to the invariance of the integrand. Thence, the sum over $\rho$ yields a factor $n!$.

In~\eqref{3.3.3a}, the differential operator $\Delta(-x^\nu\partial_{x}x^{1-\nu}\partial_x)$ can be pulled outside the inverse Hankel transfrom which is the $s$-integral, because the remaining integral is bounded by a constant times $(\prod_{j=1}^n(1+s_j)^{\nu/2-1/4})\f f_\diag(\sqrt{y_1},\ldots,\sqrt{y_n})$ which is evidently absolutely integrable due to the differentiability of $f_\diag$. This differentiability implies that $ \f f_\diag(2\sqrt{s_1},\ldots,2\sqrt{s_n})$ is bounded by a constant times $\prod_{j=1}^n(1+s_j)^{-n-\nu/2-1/2}$ for $s\geq1$.

As the proof of the following proposition works exactly along the same lines as for Proposition~\ref{p3.2.1}, apart from the subtle difference we have pointed out, we omit it here.

\begin{proposition}[Derivative principle for $\Mat(n,n+\nu)$]\label{p3.4.1}
	Let $X\in\Mat(n,n+\nu)$ be a $\U(n)\times \U(n+\nu)$-invariant random matrix with a joint squared singular value distribution $f\in L^1(\R^n)$ and a joint pseudo-diagonal entry distribution $f_\diag\in L^{1,n(n+\nu+1)}_\f(\R^n)$. These two joint distributions satisfy the relation
	\begin{equation}\label{3.3.2}
	f(x)=\frac{1}{n!C_\nu}\Delta(x)\Delta(-x^\nu\partial_xx^{1-\nu}\partial_x)\,\mathcal{A}_\nu^{-1} f_\diag(\sqrt{x})
	\end{equation}
	for almost all $x\in\mathbb{R}_+^n$.
\end{proposition}

The derivative principle~\eqref{3.3.2} as well as~\eqref{3.2.13a} for the cases $\nu=\pm1/2$ have again a direct implication for the sum of two invariant matrices $C=A+B$. The pseudo-diagonal matrix entries evidently add up so that their joint distributions $f_\diag^{(A)}$ and $f_\diag^{(B)}$ experience a simple additive convolution of two $n$-dimensional vectors. In combination with the derivative principles we can conclude the following corollary.

\begin{corollary}[Additive convolution on $M_\nu$]\label{corr.addMnu}
Let $\nu\in\mathbb{N}_0\cup\{\pm1/2\}$ and $A,B\in M_\nu$ be two independent invariant random matrices. Additionally, we assume that their joint distributions $f_\diag^{(A)}, f_\diag^{(B)}$ of their pseudo-diagonal elements satisfy the requirements of Proposition~\ref{p3.2.1} or~\ref{p3.4.1}, respectively. The joint probability density $f_C$ of the squared singular values $x$ of $C=A+B$ is then
\begin{equation}
f_C(x)=\frac{1}{n!C_\nu}\Delta(x)\Delta(-x^\nu\partial_xx^{1-\nu}\partial_x)\,\mathcal{A}_\nu^{-1}[f_\diag^{(A)}\ast f_\diag^{(B)}](\sqrt{x}).
\end{equation}
\end{corollary}

In analogy to the derivative \textit{principle} for invariant Hermitian matrices, especially Corollary~\ref{p3.1.3}, we would like to conclude this section with a uniqueness statement for the function $f_{\rm diag}$. The ensuing considerations apply for both, the present section as well as section~\ref{s3.2}. For this aim, we first give the counterpart of Lemma~\ref{p3.1.2}.

\begin{lemma}\label{p3.4.3}
	Let $u\in L^{1,n(n-1)}_\f(\R^n)$ be a function that is even in each entry and $\nu\in\mathbb{N}_0\cup\{\pm 1/2\}$. Moreover, we consider a polynomial $P(x^\nu\partial_{x}x^{1-\nu}\partial_x)\not\equiv 0$ of the partial derivatives $x_1^\nu\partial_{x_1}x_1^{1-\nu}\partial_{x_1},\ldots,x_n^\nu\partial_{x_n}x_n^{1-\nu}\partial_{x_n}$. Then, the partial differential equation
	\begin{equation}\label{3.5.1}
	P(x^\nu\partial_{x}x^{1-\nu}\partial_x)u(x)=0\ {\rm for\ all}\ x\in\mathbb{R}_+^n
	\end{equation}
	gives the unique solution $u(x)=0$ for all $x\in\mathbb{R}^n$.
\end{lemma}

\begin{proof}
	The proof is essentially the same as the one for Lemma~\ref{p3.1.2} where we now apply the Hankel transform on~\eqref{3.5.1}. This transforms $P(x^\nu\partial_{x}x^{1-\nu}\partial_x)$ to $P(s)$ due to~\eqref{2.2.5b}. In this way, one can show that $\h_\nu u(s)=0$ for all $s$ with $P(s)\neq0$ and since $P$ has been a non-zero polynomial one can extend the result to $\h_\nu u(s)=0$ for all $s\in\mathbb{R}_+^n$. The injectivity of the Hankel transform on $L^{1,n(n-1)}_\f(\R_+^n)$ leads to $u(x)=0$ for all $x\in\mathbb{R}_+^n$, and the symmetry of $u$ extents it to $x\in\mathbb{R}^n$.
\end{proof}

\begin{corollary}[Uniqueness of additive invariant ensembles on $M_\nu$]\label{p3.4.4}
	Considering the setting of Proposition~\eqref{p3.2.1} for $\nu=\pm1/2$ and of Proposition~\ref{p3.4.1} for $\nu\in\mathbb{N}_0$. Then there exists a unique permutation invariant function $w\in L^{1,n(n-1)}_\f(\R^n)$ which is even in each argument such that the joint probability distribution of the squared singular values $x$ of the invariant random matrix $X\in M_\nu$ is
	\begin{equation}
	\begin{split}
	f(x)=
	\frac{1}{n!C_{\nu}}\Delta(x)\Delta(-x^\nu\partial_xx^{1-\nu}\partial_x)w(x).
	\end{split}
	\end{equation}
	Moreover, it is $w(x)=\mathcal{A}_\nu^{-1}f_\diag(\sqrt{x})$.
\end{corollary}

Let us underline that the representation~\eqref{2.2.8a} for the Abel transform only works for $\nu\in\mathbb{N}_0$. For the two cases $\nu=\pm1/2$ it is given by $\mathcal{A}_{-1/2}^{-1}f_\diag(\sqrt{x})=\pi^{n/2}f_\diag(\sqrt{x})/\sqrt{\prod_{j=1}^nx_j}$ and $\mathcal{A}_{1/2}^{-1}f_\diag(\sqrt{x})=\pi^{n/2}(\prod_{j=1}^n-\sqrt{x_j}\partial_{x_j})f_\diag(\sqrt{x})$.

\subsection{Relationship with Rossmann's theorem}\label{s3.4}

The Harish-Chandra class is a class of Lie groups used in representation theory. In Rossmann's article~\cite{Rossmann} a relationship of Fourier transforms within this class is proposed. Let $G$ be a Lie group of the Harish-Chandra class, $\mathfrak g$ be its Lie algebra and $\mathfrak t$ be the Cartan subalgebra of $\mathfrak g$. Fix a root system of positive roots for the complexification $\mathfrak t_\C$ of $\mathfrak t$, and denote $\pi:\mathfrak t_\C\mapsto \C$ the product of those positive roots. For an invariant function $\phi$ defined on $\mathfrak g$ (i.e. $f(\mathrm{Ad}(g)x)=f(x)$ where $\mathrm{Ad}$ is the Adjoint representation of the Lie group), Rossmann's theorem implies the following equation
\begin{equation}\label{3.4.1}
	\pi(t)\int_\mathfrak g e^{i(t,y)}f(y)\dv y=\int_\mathfrak t e^{i(t,s)}\pi(s)f(s)\dv s,\quad t\in\mathfrak t.
\end{equation}
Also, the same equation is implicitly implied in Harish-Chandra's work (see e.g.~\cite{Frenkel}).

This equation \eqref{3.4.1} reduces to our \eqref{3.1.12a}, \eqref{3.2.4} and the respective formula in the $i \usp(2n)$ case. Upon a symmetrisation over the Weyl group $W$, equation \eqref{3.4.1} multiplied with a suitable scalar can be rewritten as
\begin{equation}\label{3.4.2}
	\int_\mathfrak g e^{i(t,x)}f(x)\dv x= \int_\mathfrak t \phi(t,s)f(s)\pi(s)^2\dv s,
\end{equation}
where $\phi$ is Harish-Chandra's general spherical function. Then by the Weyl integration formula, the right hand side of \eqref{3.4.2} gives exactly a spherical transform of the random matrix with matrix distribution $f$, when the Lie algebra is taken as $\mathfrak u(n), \o(2n), \o(2n+1), \usp(2n)$. The left hand side of \eqref{3.4.2} is considered to be the matrix Fourier transform of the same random matrix with only an element $t$ from the Cartan subalgebra $\mathfrak t$. The inner product there is identified as $\tr xi(h)$, where $i$ is a natural embedding of the Cartan subalgebra into the Lie algebra $\mathfrak g$, which also coincides with our notions of diagonal or pseudo-diagonal entries. Hence, the matrix Fourier transform carries over to the Fourier transforms on Cartan sub-algebras.

This discussion shows it is promising to not only unify our discoveries in those previous additive ensembles, but also make a more general statement for Lie groups in the Harish-Chandra class. This open question certainly suggests further studies while it is beyond the scope of this present work. Finally, let us underscore that the case of the set $\mathrm{Mat}(n,n+\nu)$ case is not included in Rosmann's theorem, as it is not a Lie algebra but rather a symmetric space.

\section{Multiplicative unitarily invariant ensembles}\label{s4}

\subsection{Derivative principle on the multiplicative space $\herm_+(n)$}\label{s4.1}

In the proofs of the previous additive cases, an essential fact is that the spherical transform of the random matrix is equal to a multivariate Fourier transform of some additive variables. Those variables were the diagonal entries in the cases of $\herm(n)$ and $i\usp(2n)$ and they were the pseudo-diagonal entries for the remaining matrix spaces $M_\nu$.
The question that has to be solved is which kind of multivariate transform is equal to spherical transform in the multiplicative setting.

To answer this question for the multiplicative $\herm_+(n)$ case, we turn to the LU decomposition of a matrix. Indeed, any positive definite matrix $X=LU$ can be decomposed into a lower triangular matrix $L$ with $1$'s on its diagonal and an upper triangular matrix $U$ with diagonal entries $u=(u_{1,1},\ldots,u_{n,n})\in\mathbb{R}_+^n$ being all positive. If $X$ has been Hermitian, this decomposition reduces to an equivalent representation of the Cholesky decomposition (see e.g.~\cite{Fo10}). The change of the measure becomes
\begin{equation}\label{4.1.1}
\dv X=\prod_{1\le j\le k\le n}\dv u_{j,k},
\end{equation}
where $u_{j,k}$ is the $(j,k)$-th entry of $U$. This measure is independent of the matrix entries of $L$ as the off-diagonal entries satisfy the relation $u_{j,j}^*l_{k,j}=u_{j,k}^*$ with $(.)^*$ the complex conjugation.

The LU decomposition enables us to compute the principal minors (determinants of the top left sub-block) that play a crucial role in the spherical transform~\eqref{2.3.2}. This becomes transparent in
\begin{equation}
\det X_{j\times j}=\det L_{j\times n}U_{n\times j}=\det L_{j\times j}\det U_{j\times j}=\prod_{k=1}^j u_{j,j}.
\end{equation}
So the generalised power~\eqref{2.3.1} can be be made explicit as follows
\begin{equation}
|X|^s=\prod_{j=1}^{n-1}\left(\prod_{k=1}^ju_{k,k}^{s_j-s_{j+1}-1}\right)\prod_{k=1}^nu_{k,k}^{s_n}=\prod_{j=1}^{n}u_{j,j}^{s_j-n+j}.
\end{equation}
This expression has two consequences. Firstly, the spherical transform~\eqref{2.3.2} is simply an average of a function of powers of the $u_{j,j}$'s which is on the other hand the multivariate Mellin transform in these quantities. Secondly, the spherical transform does not care about the distribution of the other entries of the matrix $U$. Thus, we can exploit the measure~\eqref{4.1.1} and integrate over all off-diagonal entries of the $U$ matrix to obtain the marginal distribution of the diagonal entries of $U$,
\begin{equation}\label{4.1.4a}
f_{u}(u):=\int F(LU)\prod_{1\le j<k\le n}\dv u_{j,k}.
\end{equation} 
The problem we encounter is that $f_u$ is not permutation invariant in its arguments. This is only true for the function
\begin{equation}\label{4.1.5}
g(u):=f_u(u)\prod_{j=1}^n u_{j}^{-n+j},
\end{equation} 
which is quite natural as aforementioned equality of the Mellin transform with the spherical transform already exhibits this product. Particularly, we have
\begin{equation}
\mathcal{S} f(s)=\m f_u(s_1-2n+2,\ldots,s_n-2n+n+1)=\m g(s_1-n+1,\ldots,s_n-n+1)
\end{equation}
when $f\in L^1(\mathbb{R}_+^n)$ is the joint distribution of the eigenvalues of $X\in\herm_+(n)$. This formula can be also seen as a consequence of~\cite[Lemma 5.3]{KK19}. 

 As before and as a last preparation, we need a regularity condition to guarantee the derivatives which is with $g\in L^{1,n(n-1)/2}_{\m}(\R^n_+)$.

\begin{proposition}[Derivative principle for $\herm_+(n)$]\label{prop:hermplus}
	Let $X\in\herm_+(n)$ be a $\U(n)$-invariant matrix with eigenvalue distribution $f\in L^1(\R^n_+)$ and function $g\in L^{1,n(n-1)/2}_{\m}(\R^n_+)$ as defined in~\eqref{4.1.5}. Then, the derivative principle reads
	\begin{equation}\label{4.1.6}
	f(x)=\frac{1}{\prod_{j=0}^{n}j!}\Delta(x)\Delta(-x\partial_x)g(x)
	\end{equation}
	for almost all $x\in\mathbb{R}_+^n$.
\end{proposition}

Comparing this proposition and its requirements with the derivative principle of the additive $\herm(n)$, see Proposition~\ref{p3.1.1}, shows very strong similarities. Indeed, a change of variables $x\to e^{\tilde{x}/L}$ in the limit $L\to\infty$ reduces Eq.~\eqref{4.1.6} to~\eqref{3.1.5}.

\begin{proof}
	We start from
	\begin{equation}\label{4.1.4}
	\s f(s)=\int_{\R_+^n} \dv u\,f_u(u)\prod_{j=1}^{n}u_{j,j}^{s_j-2n+j}=\m g(s_1-n+1,\ldots,s_n-n+1).
	\end{equation}
	Note that the left hand side of~\eqref{4.1.4} is permutation invariant in $s$ by the definition of the spherical transform, and so $g$ is also a permutation invariant function.
	
	When we substitute Eq.~\eqref{4.1.4} into the inverse spherical transform~\eqref{2.3.5}, we can expand the determinant $\det[x_j^{-s_k}]_{j,k=1}^{n}$ and obtain
	\begin{equation}
	f(x)=\frac{\Delta(x)}{\prod_{j=0}^{n}j!}\frac{1}{n!}\sum_{\rho\in\mathrm{S}_n}\int_{\R^n}\frac{\dv s}{(2\pi)^n}\Delta(is_1,\ldots,n-1+is_n)\sgn(\rho)\prod_{j=1}^nx_{\rho(j)}^{n-j-is_j}\m g(1-n+is_1,\ldots,is_n).
	\end{equation}
	Next, we rewrite
	\begin{equation}
	\begin{split}
	\Delta(is_1,&\ldots,n-1+is_n)\sgn(\rho)\prod_{j=1}^nx_{\rho(j)}^{n-j-is_j}
	\\&=\Delta(1-n+is_1,\ldots,is_n)\sgn(\rho)\prod_{j=1}^nx_{\rho(j)}^{n-j-is_j}
	=\Delta(-x\partial_x)\prod_{j=1}^nx_{\rho(j)}^{n-j-is_j}.
	\end{split}
	\end{equation}
	For the first equality we have employed the translation invariance of the Vandermonde determinant. To interchange this derivative operator with the integral we need to exploit the regularity condition of $\m g$ which renders the remaining integrand absolutely integrable. This remaining integral is the inverse multivariate Mellin transform which yields $g(x_{\rho(1)},\ldots,x_{\rho(n)})$ for almost all $x\in\mathbb{R}_+^n$. Due to the permutation invariance of $g$ the sum of $\rho$ is trivial and gives a factor $n!$ which concludes the proof.
\end{proof}

By a similar argument to Lemma~\ref{p3.1.2}, one can show that for $u\in L^{1,n(n-1)/2}_{\m}(\R_+^n)$, the differential equation $\Delta(-x\partial_x)u(x)=0$ has a unique solution $u\equiv 0$. This shows the uniqueness of the expression~\eqref{4.1.6} for all invariant ensembles on $\herm_+(n)$. We summarise this in the following corollary.

\begin{corollary}[Uniqueness of multiplicative invariant ensembles on $\herm(n)$]\label{p4.1.3}
	Considering the situation of Proposition~\ref{prop:hermplus}, the function $ w\in L^{1,n(n-1)/2}_{\m}(\R_+^n)$ satisfying
	\begin{equation}\label{4.1.8}
	f(x)=\frac{1}{\prod_{j=0}^nj!}\Delta(x)\Delta(-x\partial_x)w(x)
	\end{equation}
	is unique.
	We refer $w$ as the multiplicative weight, and it satisfies the relation
	\begin{equation}\label{4.1.9}
	\s f(s)=\m w(s_1-n+1,\ldots,s_n-n+1),
	\end{equation}
	which implies $w=g$.
\end{corollary}

\begin{example}[Examples for invariant ensembles on $\herm_+(n)$]\label{p4.1.4}$\,$
	\begin{enumerate}
		\item (Polynomial ensemble) When we choose the notation of case 3 in Example~\ref{p3.1.4}, we can easily derive the counterpart of~\eqref{polya-add}. It is
		\begin{equation}
		w(x)=\m^{-1}\left[\frac{\det\left[\m w_k(s_j) \right]_{j,k=1}^n}{\Delta(s)}\right](x).
		\end{equation}
		As before this integral can be quite complicated if the two determinants in the inverse multivariate Mellin transform do not simplify as it is for the multiplicative P\'olya ensembles the case.
		\item (Multiplicative P\'olya ensemble~\cite{KK19,FKK17}) A multiplicative P\'olya ensemble on $\herm_+(n)$ is another subclass of polynomial ensembles and its intersection with the additive P\'olya ensembles on $\herm(n)$ is not very big. While the Wishart-Laguerre ensemble is contained in both, the additive and multiplicative P\'olya ensembles, see~\cite{KK19}, this is not the case for the Jacobi ensemble which is only a multiplicative P\'olya ensemble.
		
		Also this time the multiplicative weight function $w$ factorises $w(x)=\prod_{j=1}^n\tilde w(x_j)$. In~\cite[Corollary 5.4]{KK19} it is shown that $w$ can be decomposed in such a way only when the diagonal entries of the $U$ matrix are independent, i.e. $f_u$ is decomposed, which coincides with our result. The crucial and only formal difference to the additive P\'olya ensemble is the derivative operator that is applied to the weight $w$.
	\end{enumerate}
\end{example}

\subsection{Derivative principle on $\U(n)$}\label{s4.2}

When carrying over the discussion for the multiplication on $\herm_+(n)$ to $\U(n)$, it is necessary to make use of an LU decomposition for unitary matrices in order to express the principal minors $\det X_{j\times j}$ properly. In Appendix~\ref{ab}, we give a suitable parametrisation of a matrix $X\in \U(n)$ in terms of the following variables:
\begin{enumerate}
\item the radii of the first $n-1$ principal minors $r_l, (l=1,\ldots,n-1)$,
\item the angles of the phases of all principal minors $\varphi_l, (l=1,\ldots,n)$,
\item and $(n-1)^2$ remaining angles $\theta_{j,k}, (2\leq j+1<k\leq n)$ and $\psi_{j,k}, (1\leq j<k\leq n)$.
\end{enumerate}
They are drawn from the following range
\begin{equation}
0<r_{l}<R_l,\quad -\pi<\varphi_l<\pi,\quad 0<\phi_{j,k},\psi_{j,k}<\pi/2,
\end{equation}
where $R_j$ are functions given by
\begin{equation}
R_j=R_j(\Theta)=\frac{\prod_{l=1}^j\prod_{k=j+1}^n\cos\phi_{l,k}}{\cos\phi_{j,j+1}}<1.
\end{equation}
Out of convenience, we introduce the abbreviation $\Theta$ for the set of all angles $\phi_{j,k}$ and $\psi_{j,k}$. The normalised Haar measure of $\mu(dX)$ has been computed in Proposition~\ref{prop:Haar.coord} and reads
\begin{equation}\label{haarmeasure}
\mu(\dv X)= \prod_{l=1}^{n-1}r_l\dv r_l\prod_{l=1}^{n}\dv\varphi_l\,\dv \Theta,
\end{equation} 
where $\dv \Theta$ is a measure given explicitly by
\begin{equation}
\dv\Theta=\left(\prod_{k=1}^n\frac{(k-1)!}{2\pi^k}\right)\prod_{\substack{2\leq j+1<k\leq n}}\tan\phi_{j,k}\dv\phi_{j,k}\prod_{1\leq j<k\leq n}\dv\psi_{j,k}.
\end{equation}

The problem of deriving a derivative principle for $\U(n)$ is that the principal minors of the matrix $X$ describe a two-dimensional space, meaning it is generically a complex number inside the unit ball, in contrast to the case $\herm_+(n)$ where the principal minors have been still positive numbers. This is certainly reflected in the derivative principle for the distribution of the eigenangles $\theta=(\theta_1,\ldots,\theta_n)\in(-\pi,\pi]^n$ of the invariant random matrix $X\in\U(n)$. Let us go over to an equivalent set of angles $\zeta_{\rho}^{(l)}=\sum_{j=1}^l\theta_{\rho(j)}$. Then, we need to introduce a function
\begin{equation}\label{4.2.5}
g(\theta):=\sum_{\rho\in\mathrm{S}_n}\left\langle\delta\left(\varphi_n-\zeta_{\rho}^{(n)}\right)\prod_{l=1}^{n-1}\frac{1}{2\pi}\frac{e^{-i\zeta_{\rho}^{(l)}}}{1-r_le^{i(\varphi_l-\zeta_{\rho}^{(l)})}}\right\rangle
\end{equation}
for stating the desired derivative principle. The angular bracket represents the ensemble average over $X\in\U(n)$. The Dirac delta function $\delta\left(\varphi_n-\zeta_{\rho}^{(n)}\right)$ has to be understood as $2\pi$-periodic to guarantee the proper symmetry of $g$.

\begin{proposition}[Derivative principle on $\U(n)$]\label{p4.4}
	Let $X\in\U(n)$ be a $\U(n)$-invariant matrix with a joint eigenvalue distribution $f\in L^1((-\pi,\pi]^n)$. Defining the function $g\in L_{\f}^{1,n(n-1)/2}((-\pi,\pi]^n)$ as in~\eqref{4.2.5}, we find the derivative principle relating $f$ and $g$
	\begin{equation}\label{4.2.6}
	f(e^{i\theta})=\frac{1}{\prod_{j=0}^{n}j!}\Delta(e^{i\theta})\Delta(i\partial_\theta)g(\theta).
	\end{equation}
\end{proposition}

\begin{proof}
	Our starting point is the explicit form of the spherical transform of $X$ in terms of our parametrisation introduced in Appendix~\ref{ab},
	\begin{equation}
	\begin{split}\label{4.2.7}
	\s f(s)=&\left\langle \prod_{j=1}^n\det(X_{j\times j})^{s_j-s_{j+1}-1}\det X^{s_n}\right\rangle=\left\langle \prod_{j=1}^n(r_je^{i\varphi_j})^{s_j-s_{j+1}-1}(e^{i\varphi_n})^{s_n}\right\rangle.
	\end{split}
	\end{equation}
	We would like to underline that this formula is only valid for $s_1>s_2>\ldots>s_{n}$ as otherwise the average might run through poles. For the other orders we exploit the permutation invariance of the spherical transform. Hence the whole spherical transform is
	\begin{equation}
	\begin{split}\label{4.2.7b}
	\s f(s)=&\sum_{\rho\in\mathrm{S}_n}\left\langle \prod_{j=1}^n(r_je^{i\varphi_j})^{s_{\rho(j)}-s_{\rho(j+1)}-1}(e^{i\varphi_n})^{s_{\rho(n)}}\right\rangle\prod_{j=1}^{n-1}H(s_{\rho(j)}-s_{\rho(j+1)})
	\end{split}
	\end{equation}
	with the aid of the Heaviside step function $H(s_{\rho (j)}-s_{\rho (j+1)})$ which is only one when $s_{\rho (j)}-s_{\rho (j+1)}>0$ and otherwise vanishes.

	We compare the expression~\eqref{4.2.7b} with the multivariate Fourier transform of
	\begin{equation}\label{4.2.5b}
	\begin{split}
		g(\theta)=\sum_{\rho\in\mathrm{S}_n}\left\langle\delta\left(\varphi_n-\zeta_{\rho}^{(n)}\right)\prod_{l=1}^{n-1}\frac{1}{2\pi}\sum_{\tilde z_l=1}^\infty (r_le^{i\varphi_l})^{\tilde z_l-1}e^{-i\tilde z_l\zeta_{\rho}^{(l)}}\right\rangle
	\end{split}
	\end{equation}
	where	$\zeta_{\rho}^{(l)}=\sum_{j=1}^l\theta_{\rho(j)}$. In this equation, we have identified the geometric series as $r_l\leq1$ where the event $r_l=1$ is only of measure zero. The Fourier transform can be computed as follows
	\begin{equation}\label{4.2.5c}
	\begin{split}
	\f g(s)=&\int_{-\pi}^\pi \dv\theta_1 e^{i\theta_1 s_1}\cdots\int_{-\pi}^\pi \dv\theta_n e^{i\theta_n s_n}\sum_{\rho\in\mathrm{S}_n}\left\langle\delta\left(\varphi_n-\zeta_{\rho}^{(n)}\right)\prod_{l=1}^{n-1}\frac{1}{2\pi}\sum_{\tilde z_l=1}^\infty (r_le^{i\varphi_l})^{\tilde z_l-1}e^{-i\tilde z_l\zeta_{\rho}^{(l)}}\right\rangle \\
	=&\sum_{\rho\in\mathrm{S}_n}\int_{-\pi}^\pi \frac{\dv\zeta_{\rho}^{(1)}}{2\pi}\cdots\int_{-\pi}^\pi \frac{\dv\zeta_{\rho}^{(n-1)}}{2\pi}\int_{-\pi}^\pi \dv\zeta_{\rho}^{(n)}\\
	&\times\left\langle e^{i\zeta_{\rho}^{(n)}s_{\rho(n)}}\delta\left(\varphi_n-\zeta_{\rho}^{(n)}\right)\prod_{l=1}^{n-1}\sum_{\tilde z_l=1}^\infty (r_le^{i\varphi_l})^{\tilde z_l-1}e^{i(s_{\rho(l)}-s_{\rho(l+1)}-\tilde z_l)\zeta_{\rho}^{(l)}}\right\rangle \\
	=&\s f(s).
	\end{split}
	\end{equation}
	In the second equality, we have changed the variables $\theta_{\rho(l)}=\zeta_{\rho}^{(l)}-\zeta_{\rho}^{(l-1)}$ with $\zeta_{\rho}^{(1)}=\theta_{\rho(1)}$. The new variables are still $n$ independent angles that run from $-\pi$ to $\pi$ exploiting the $2\pi$ periodicity.

	Equation~\eqref{4.2.5c} is at the core of the derivative principle which can then be derived after applying the inverse spherical transform,
	\begin{equation}
	\begin{split}
	f(e^{i\theta})
	&=\frac{\Delta(e^{i\theta})}{(2\pi)^nn!\prod_{j=0}^nj!}\sum_{s\in\Z^n}\sum_{\rho\in\mathrm{S}_n}\Delta(s)\sgn(\rho)\s f(s)\prod_{j=1}^ne^{-is_j\theta_{\rho(j)}}\\
	&=\frac{\Delta(e^{i\theta})}{(2\pi)^nn!\prod_{j=0}^nj!}\sum_{s\in\Z^n}\sum_{\rho\in\mathrm{S}_n}\Delta(s)\sgn(\rho)\f g(s)\prod_{j=1}^ne^{-is_j\theta_{\rho(j)}}\\
	&=\frac{\Delta(e^{i\theta})\Delta(i\partial_\theta)}{(2\pi)^nn!\prod_{j=0}^nj!}\sum_{\substack{s\in\Z^n\\s_j\ne s_k,\forall j,k}}\sum_{\rho\in\mathrm{S}_n}\f g(s)\prod_{j=1}^ne^{-is_j\theta_{\rho(j)}}.\label{4.2.9}
	\end{split}
	\end{equation}
	The cases $s_j=s_k$ can be excluded from the $s$-sum in the second line since their contributions are zero due to the Vandermonde determinant in $s$ as well as the missing Fourier coefficients of $g$. In the third line, we have re-expressed $\Delta(s)\sgn(\rho)$ in terms of the derivative operator $\Delta(i\partial_\theta)$ which will be pulled outside the series. As before this interchange of the series and the derivative is guaranteed by dominated convergence because the differentiability condition of $g\in L_{\f}^{1,n(n-1)/2}((-\pi,\pi]^n)$ leads to an absolutely convergent series of its Fourier transform. The sum over $s$ is equal to the inverse Fourier transform of $g$ and the sum over $\rho$ is trivial as $g$ is permutation invariant so that it yields a factor of $n!$. This finishes the proof.
\end{proof}

As before, we can ask how unique the representation~\eqref{4.2.6} is. Unfortunately, $\U(n)$ invariant ensembles do not allow a simple analogue of Propositions~\ref{p3.1.1},~\ref{p3.4.1} and Corollary~\ref{p4.1.3}. The simple reason is that solutions of the partial differential equation
\begin{equation}
	\Delta(\partial_\theta)u(\theta)=0
\end{equation}
for $u$ being a $2\pi$-periodic function allow the constant since it is still integrable on the compact domain $(-\pi,\pi]^n$. This is the reason why also any $2\pi$ periodic function which only depends on the sum of the angles $\sum_{j=1}^n\theta_j$ admits this differential equation. Those solutions were suppressed in the previous cases because they were not integrable.
However, Proposition~\ref{p4.4} suggests the following modified version of a uniqueness statement, such that multiplicative weight is unique if those Fourier coefficients vanish whenever $s_j=s_k$ for some $j,k\in\{1,\ldots,n\}$. Indeed this is a rather straightforward consequence when noticing that these points drop out due to the Vandermonde determinant $\Delta(s)$ in the series which have been carried out in~\eqref{4.2.9}. The rest follows from the injectivity of the spherical transform so that we can omit its proof.

\begin{corollary}[Uniqueness of invariant ensembles on $\U(n)$]
	Considering the requirements of Proposition~\ref{p4.4} for an invariant ensemble on $\U(n)$, its joint eigenvalue distribution
	\begin{equation}\label{4.2.14}
	f(e^{i\theta})=\frac{1}{\prod_{j=0}^nj!}\Delta(e^{i\theta})\Delta(i\partial_\theta)w(\theta),
	\end{equation}
	corresponds uniquely to a $2\pi$-periodic function $w\in L_{\f}^{1,n(n-1)/2}((-\pi,\pi]^n)$ if one furthermore requires its multivariate Fourier transform to satisfy
	\begin{equation}
	\f w(s)=0\ {\rm whenever}\ s_j=s_k\ {\rm for\ some}\ j,k\in\{1,\ldots,n\}.
	\end{equation}
	Moreover, we have $w=g$ and, thence, the multiplicative weight $w$ admits the relation
	\begin{equation}\label{4.2.16}
	\s f(s)=\f w(s).
	\end{equation}
\end{corollary}

As an example, we would like to conclude the section with computing $g(\theta)$ for the circular unitary ensemble meaning Haar distributed unitary matrices. It has the simple matrix distribution $F(X)\equiv 1$.

\begin{proposition}[Multiplicative weight of CUE]\label{p4.2.2} For the circular unitary ensemble, the multiplicative weight $g$ is explicitly given by
	\begin{equation}
	g(\theta)=\frac{1}{(2\pi)^n}\sum_{\rho\in\mathrm{S}_n}\prod_{l=1}^{n-1}e^{-i(\theta_{\rho(1)}+\ldots+\theta_{\rho(l)})}=\frac{1}{(2\pi)^n}\mathrm{perm}[e^{-i(j-1)\theta_k}]_{j,k=1}^{n},
	\end{equation}
	where $\mathrm{perm}$ specifies the permanent of a matrix. This yields an alternative expression for the joint probability density of the eigenvalues which is
	\begin{equation}
	f(e^{i\theta})=\frac{1}{(2\pi)^n\prod_{j=0}^nj!}\Delta(e^{i\theta})\,\Delta(i\partial_\theta)\,\mathrm{perm}[e^{-i(j-1)\theta}]_{j,k=1}^n.
	\end{equation}
\end{proposition}

\begin{proof}
	The angles $\varphi_l\in(-\pi,\pi]$ come with a flat measure, cf., Eq.~\eqref{haarmeasure}. Therefore, the integral over $\varphi_n$ is trivial as it evaluates the Dirac delta function yielding a factor of $1/(2\pi)$ which has been the normalisation of this integration variable. The other integrals in $\varphi_l$ are given by
	\begin{equation}
	\int_{-\pi}^\pi\frac{\dv \varphi}{1-r e^{i(\varphi-\zeta)}}=\int_{-\pi}^\pi \dv \varphi\left(1+\sum_{z=1}^\infty re^{i(\varphi-\zeta)z}\right)=2\pi
	\end{equation}
	because $r\leq1$ and the subsets described by $r=1$ are of measure zero.
	When we substitute this into~\eqref{4.2.5}, we obtain
	\begin{equation}
	\begin{split}
	g(\theta)&=\frac{1}{2\pi}\left\langle\int_{-\pi}^\pi\frac{\dv\varphi_1}{2\pi}\ldots\int_{-\pi}^\pi\frac{\dv\varphi_{n-1}}{2\pi}\prod_{l=1}^{n-1}\frac{1}{2\pi}\frac{e^{-i\zeta_{\rho}^{(l)}}}{1-r_le^{i(\varphi_l-\zeta_{\rho}^{(l)})}}\right\rangle_{\Theta,r}\\
	=&\frac{\langle 1\rangle_{\Theta,r}}{(2\pi)^{n}}\sum_{\rho\in\mathrm{S}_n}\prod_{l=1}^{n-1}e^{-i\zeta_{\rho}^{(l)}}\\
	=&\frac{1}{(2\pi)^n}\sum_{\rho\in\mathrm{S}_n}\prod_{l=1}^{n}e^{-i(n-j)\theta_{\rho(j)}},
	\end{split}
	\end{equation}
	where $\langle\cdot\rangle_{\Theta,r}$ denotes the $\Theta$ and $r_l$-integrals in the ensemble average. Clearly, $\langle 1\rangle_{\Theta,r}$ equals to $1$, and, in the last step we have extended the product to $l=n$ as $e^{-i(n-n)\theta_{\rho(n)}}=1$. The sum over $\rho$ is exactly the definition of the permanent, which concludes the proof.
\end{proof}

It is quite interesting and surprising that the multiplicative weight of a Haar distributed unitary random matrix is equal to the permanent and not the product of some weights as we have seen for P\'olya ensembles on $\herm_+(n)$. We defer the discussion of this problem to future work.

\section{CONCLUDING REMARKS}\label{s5}

In the present work, we generalised the derivative principle on $\herm(n)$ matrices to the matrix spaces $M_\nu=\{i\o(n),i\usp(2n),{\rm Mat}(n,n+\nu)\}$, $ \herm_+(n)$ and $\U(n)$. Intriguingly, the differential operators involved are also found for P\'olya ensembles~\cite{FKK17,Ki19,KFI19,KK16,KK19,KR19,ZKF19} on the respective matrix spaces. Here, we would like to point out that the theory for P\'olya ensembles on $\U(n)$ has not been done, yet, but the derivative principle suggests a way how to do it. A work considering this problem is currently in preparation.

As we have shown, each of the derivative principles uniquely link the joint eigenvalue/singular value statistics with another quantity of the random matrix for which classical multivariate probability theory applies. For the additive matrix spaces, these quantities exhibit the nice property that they follow simple additive convolutions when adding two independent random matrices. This allows for further studies like central limit theorems of invariant random matrix ensembles, but also computations of quantities like the level density or similar observables for general ensembles might be possible.

To briefly summarise our proofs, we made use of the HCIZ-type integrals (HCIZ~\eqref{2.2.3}, Berezin-Karpelevich~\eqref{2.2.6}, and Gelfand-Na\u\i mark~\eqref{2.3.3}) in those matrix spaces. In this way, we could show that the respective spherical transforms of the joint eigenvalue/singular value distribution agrees with the multivariate transforms of other matrix quantities such as the Fourier or Hankel transform of the joint probability distribution of the diagonal or pseudo-diagonal matrix entries. We believe that this method can be carried over to other HCIZ-type integrals and similar discussions can be made. Recent studies of harmonic analysis on other matrix spaces are being discussed in~\cite{KFI19}, and a future work of this paper is to generalise our idea to those matrix spaces.

In addition, in Appendix~\ref{ab} we presented a new way of parametrising unitary matrices, which is suitable for LU decompositions. This parametrisation is inspired by Hurwitz' parametrisation, see~\cite{Hu97,DF17}. In~\cite{DF17} similar parametrisations for orthogonal and symplectic groups have been given. Thus, it is also natural to ask for a generalisation of our new parametrisation to those two compact groups as it is quite likely to find derivative principles for those sets, too.


Altland and Zirnbauer~\cite{AZ01} gave a full classification of symmetric matrix symmetric spaces of Hermitian type in terms of their Cartan symbols. What we have considered in the present work are the Lie algebras $\herm(n)$ with symbol $A$, $i\o(N)$ with symbol $BD$, $i\usp(2n)$ with symbol $C$ and the flat symmetric space $\mathrm{Mat}(n,n+\nu)$ with symbol $AIII$, as well as the compact and non-compact types of the class $A$, which are exactly $\U(n)$ and $\herm_+(n)$. Therefore an intriguing question as mentioned in Section~\ref{s3.4}, is whether we could find a unified approach to all symmetric matrix spaces. 

However, here we show that not all such symmetric matrix spaces inherit a derivative principle consisting of a finite order differential operator. Let us consider the $N\times N$ real symmetric matrices with orthogonal invariance, that is, $F(X)=F(QXQ^\top)$ for all $X\in\mathrm{Sym}(N)$ (real symmetric matrices) and$Q\in\O(N)$, the $N\times N$ orthogonal group. Denote its matrix distribution as $F$ and eigenvalue distribution as $f$. Unlike the Hermitian case, the relationship between $f$ and $F$ does not contain two Vandermonde determinants but instead the absolute value of a Vandermonde determinant:
\begin{equation}
	f(x)\propto|\Delta(x)|F(x).
\end{equation}

Therefore the first thing to do is to propose a different ansatz for the expression of the derivative principle in this case. To do so we firstly define the spherical transform of $X$ as the ordinary Fourier transform for all its upper triangular entries,
\begin{equation}
\begin{split}
	\s f(s):= \f F(S)& =\int_{\R^{n(n+1)/2}}F(X)\exp(i\tr XS)\dv X.
\end{split}
\end{equation}
Integrating over the special orthogonal group $\SO(N)$ by the invariance, one obtains an alternative expression in terms of the eigenvalues namely
\begin{equation}
	\s f(s) = \int_{\R^n}f(x)\phi(x,s)\dv x
\end{equation}
where $\phi(x,s)$ is the spherical function given by
\begin{equation}
	\phi(x,s):=\int_{\SO(n)}\exp(i\tr Q\diag(x)Q^\top\diag(s))\dv\mu(Q)
\end{equation}
with $\mu$ being the Haar measure on $\SO(n)$. It can be seen from the inverse matrix Fourier transform that after performing the same integration over $\SO(n)$, one obtains the inverse spherical transform
\begin{equation}
	f(x)=\s^{-1}[\s f](x) =c_n|\Delta(x)|\int_{\R^n}|\Delta(s)|\phi(x,-s)\s f(s).
\end{equation}
where $c_n$ is an appropriate normalisation constant depending only on the dimension $n$. 

We notice that the spherical transform of $f$ is equal to the Fourier transform of $f_\diag$. Therefore the derivative principle in this case is similar to all the previous cases:
\begin{equation}\label{5.6}
\begin{split}
f(x)&=\s^{-1}\f f_\diag(x)=c_n|\Delta(x)|\int_{\R^n}|\Delta(s)|\phi(x,-s)\f f_\diag(s).
\end{split}
\end{equation}
Equation~\eqref{5.6} has a similar form to~\eqref{3.1.7}. Hence a systematic way to rewrite this derivative principle is
\begin{equation}
	f(x)=|\Delta(x)|D_x^{(n)}f_\diag(x)
\end{equation}
where $D_x^{(n)}$ is the integral operator given by the integral in~\eqref{5.6} with the Fourier transform.

However, unlike the Hermitian case where a similar integral operator can be simplified to a Vandermonde determinant of partial derivatives, in the real symmetric case this is no longer possible. The following proposition shows that when $n\ge 2$, $D_x^{(n)}$ can not be a linear differential operator of finite order.

\begin{proposition}
	Consider a random $n\times n$ real symmetric matrices with $\O(n)$-invariance with $n\ge 2$. Denote its eigenvalue distribution as $f$ and its diagonal entry distribution as $f_\diag$. Let $D_x^{(n)}$ be an operator satisfying
	\begin{equation}\label{dp_o(n)}
		f(x)=|\Delta(x)|D_x^{(n)}f_\diag(x).
	\end{equation}
	Then $D_x^{(n)}$ can not be a finite order linear differential operator, that is, it is impossible that there exist functions $c_\pi(x)$ for each multi-index $\pi=(\pi_1,\pi_2,\ldots,\pi_n), \pi_1,\ldots,\pi_n\ge 0$ and integer $p>0$ such that
	\begin{equation}
		D_x^{(n)}=\sum_{\pi_1+\ldots+\pi_N\le p} c_\pi(x)\prod_{j=1}^n\partial_{x_j}^{\pi_j}.
	\end{equation}
\end{proposition}

\begin{proof}
	We prove by contradiction using the Laguerre orthogonal ensemble for an arbitrary $n,M\in\mathbb{N}$ as~\eqref{dp_o(n)} has to be true for an arbitrary differentiable $f_{\rm diag}$. Let $Y$ be an $M\times n$ ($M\ge n$) matrix with independent real standard Gaussian entries. Then the matrix $X=Y^\top Y$ has the eigenvalue distribution
	\begin{equation}
		f(x)=\frac{1}{Z_n}\prod_{j<k}|x_k-x_j|\prod_{j=1}^nx_j^{(M-n-1)/2}e^{-x_j/2}
	\end{equation}
	with the normalisation constant (see e.g.~\cite{Fo10})
	\begin{equation}
	\frac{1}{Z_n}=\prod_{j=1}^n\frac{\Gamma(3/2)}{2^{M/2}\Gamma((j+3)/2)\Gamma((j+M-n)/2)}.
	\end{equation}
	Also the diagonal entries are independent $\chi^2$-distributions with parameter $M$, that is
	\begin{equation}
		f_\diag(x)=\frac{1}{2^{Mn/2}\Gamma(M/2)^n}\prod_{j=1}^nx_j^{M/2-1}e^{-x_j/2}.
	\end{equation}
	
	Assuming that there exists a linear differential operator $D_x^{(N)}$ such that~\eqref{dp_o(n)} satisfies, then we should have
	\begin{equation}
	\frac{1}{Z_n}\prod_{j=1}^nx^{(M-n-1)/2}e^{-x_j/2} = \frac{1}{2^{Mn/2}\Gamma(M/2)^n}D_x^{(N)}\prod_{j=1}^nx_j^{M/2-1}e^{-x_j/2}.
	\end{equation}
	By Rodrigues' formula, the action of the partial derivatives $\partial_{x_j}^{\pi_j}$ on the Laguerre weight $x_j^{M/2-1}e^{-x_j/2}$ gives a polynomial $p_{\pi(j)}$ in both $x_j$ and $M$, and a prefactor $x^{-\pi(j)}$. Therefore,
	\begin{align}
	&\frac{1}{Z_n}\prod_{j=1}^nx^{(M-n-1)/2}e^{-x_j/2} = \frac{1}{2^{Mn/2}\Gamma(M/2)^n}\sum_{\pi_1+\ldots+\pi_N\le p} c_\pi(x)\prod_{j=1}^np_{\pi(j)}(x_j,M)x_j^{M/2-1-\pi(j)}e^{-x_j/2}\\
	\Rightarrow & \prod_{j=1}^n\frac{\Gamma(3/2)\Gamma(M/2)}{\Gamma((j+3)/2)\Gamma((M-n+j)/2)}= \sum_{\pi_1+\ldots+\pi_N\le p} c_\pi(x)\prod_{j=1}^np_{\pi(j)}(x_j,M)x_j^{n/2-1/2-\pi(j)}.\label{5.13}
	\end{align}
	Now clearly the right hand side of~\eqref{5.13} is a polynomial in $M$. However for the left hand side, for $n-j\ge 1$ each function $\Gamma(M/2)/\Gamma((M-n+j)/2)$ has a singularity at $M=0$, which is in a contradiction to the equation~\eqref{5.13}.
\end{proof}

\section*{Acknowledgements}
JZ acknowledges the support of a Melbourne postgraduate award, and an ACEMS top-up scholarship. MK acknowledges financial support from the Australian Research Council of the Discovery Project grant DP210102887. We are grateful for fruitful discussions with Peter Forrester and Shi-Hao Li.

\bibliography{DPbib.bib}{}
\bibliographystyle{plain}
\

\newpage

\appendix
\numberwithin{theorem}{section}

\section{Some proofs for Sec.~\ref{s2}}\label{aa}

\begin{proof}(Proof of Proposition~\ref{p2.2.2})

Without loss of generality we can simply consider the $n=1$ case as each individual integration in~\eqref{2.2.8a} for one diagonal element is carried out in exactly the same way.
 Essentially, we have to consider the integral
	\begin{equation}
	\h^{-1}_\nu  [\f \tilde{f}(2\sqrt{s})](x)=\int_0^\infty\dv s J_\nu(2\sqrt{xs})(xs)^{\nu/2}\int_{-\infty}^\infty\dv \lambda\,\tilde{f}(\lambda)e^{2i\lambda\sqrt{s}}.
	\end{equation}
	After changing the integration variable $s=u^2/4$ and introducing a regularisation $e^{-\varepsilon u}$ to guarantee Fubini's theorem and the interchange of variables, we obtain the following expression
	\begin{equation}
	\begin{split}
	\h^{-1}_\nu  [\f \tilde{f}(2\sqrt{s})](x)&=\frac{x^\nu}{2^{\nu+1}}\lim_{\varepsilon\rightarrow 0}\int_0^\infty\dv u\, J_\nu(\sqrt{x}u)u^{\nu+1}e^{-\varepsilon u}\int_{-\infty}^\infty\dv \lambda\,\tilde{f}(\lambda)e^{i\lambda u}
	\\
	&=\frac{x^\nu}{2^{\nu+1}}\lim_{\varepsilon\rightarrow 0}\int_{-\infty}^\infty\dv \lambda\,\int_0^\infty\dv u\, J_\nu(\sqrt{x}u)u^{\nu+1}\tilde{f}(\lambda)e^{-(\varepsilon-i\lambda)u}.
	\end{split}
	\end{equation}
	For the $u$-integral we employ the integral formula~\cite[\S 8.6 eqn.(4)]{Ba55}
	\begin{equation}
	\int_0^\infty\dv u\, J_\nu(\sqrt{x}u)u^{\nu+1}e^{-(\varepsilon-i\lambda)u}=\frac{2^{\nu+1}\Gamma(\nu+3/2)}{\sqrt{\pi}}\frac{x^{\nu/2}(\varepsilon-i\lambda)}{((\varepsilon-i\lambda)^2+x)^{\nu+3/2}}
	\end{equation}		
	 which can be applied to obtain
	\begin{equation}
	\begin{split}
	\h^{-1}_\nu [\f \tilde{f}(2\sqrt{s})](x)&=\frac{\Gamma(\nu+3/2)x^{\nu}}{\pi^{1/2}}\lim_{\varepsilon\rightarrow 0}\int_{-\infty}^\infty \dv \lambda\, \tilde f(\lambda)\frac{\varepsilon-i\lambda}{((\varepsilon-i\lambda)^2+x)^{\nu+3/2}}\\
	&=\frac{x^\nu}{2^{\nu+1}}\lim_{\varepsilon\rightarrow 0}\int_{-\infty}^{\infty}\dv \lambda\,\frac{\varepsilon-i\lambda}{\sqrt{(\varepsilon-i\lambda)^2+x}} (\lambda^{-1}\partial_\lambda)^{\nu+1} \tilde{f}(\lambda).
	\end{split}
	\end{equation}
	In the second line we have integrated by parts and omitted terms of order $\varepsilon $ and smaller. The $\varepsilon$-dependence in the factors $1/(\varepsilon-i\lambda)$ in combination with the derivatives can be neglected due to the assumed integrability and differentiability condition of the function $\tilde{f}$.
	
	The limit $\varepsilon\to0$ can be carried out almost everywhere pointwise via Lebesgue's dominated convergence theorem as the integrand is bounded by the integrable function $| (\lambda^{-1}\partial_\lambda)^{\nu+1} \tilde{f}(\lambda)|\,|\lambda|/\sqrt{|\lambda^2-x|}$ for each $\varepsilon>0$ and almost all $x>0$. We employ
	\begin{equation}
	\begin{split}
	\lim_{\varepsilon\to0}\frac{\varepsilon-i\lambda}{\sqrt{(\varepsilon-i\lambda)^2+x}}=\frac{|\lambda|}{\sqrt{\lambda^2-x}}H(\lambda^2-x)-i\frac{\lambda}{\sqrt{x-\lambda^2}}H(x-\lambda^2)
	\end{split}
	\end{equation}
	with $H$ the Heaviside step function which is $1$ for positive arguments and otherwise vanishes. Since the function $\tilde{f}$ is also symmetric about the origin, one can simplify the integral to
	\begin{equation}
	\begin{split}
	\h^{-1}_\nu  [\f \tilde{f}(2\sqrt{s})](x)=&\frac{x^{\nu}}{2^{\nu+1}}\bigg(2\int_{\sqrt{x}}^\infty \dv \lambda\, \frac{\lambda}{\sqrt{\lambda^2-x}}(\lambda^{-1}\partial_\lambda)^{\nu+1}\tilde{f}(\lambda)-i\int_{-\sqrt{x}}^{\sqrt{x}}\dv \lambda\, \frac{\lambda}{\sqrt{x-\lambda^2}}(\lambda^{-1}\partial_\lambda)^{\nu+1}\tilde{f}(\lambda)\bigg),
	\end{split}
	\end{equation}
	where the second integral vanishes because of its symmetries. Upon changing $\lambda^2=y_j$, Proposition~\ref{p2.2.2} eventually results.
\end{proof}

\begin{proof}(Proof of Proposition~\ref{p2.2.1})

	In order to show the inverse transform one needs to introduce anew a regularisation $\prod_{j=1}^ne^{-\varepsilon s_j}$ in the limit $\varepsilon\rightarrow 0^+$  so that we can interchange the integrals. For $n=2m$, by Fubini's theorem the integral in~\eqref{2.2.10} is written as
	\begin{equation}
	\begin{split}\label{3.2.9b}
	\s^{-1}[\s f](x)=&\frac{1}{(n!)^2}\lim_{\varepsilon\rightarrow 0^+}\int_{\R^n_+} \dv y \,f(y)\frac{\Delta(x)}{\Delta(y)}\prod_{j=1}^n\left(\frac{x_j}{y_j}\right)^{\nu/2}\\
	&\times \int_{\R^n_+}\dv s \det\left[J_\nu(2\sqrt{x_js_k}) \right]_{j,k=1}^n\det\left[J_\nu(2\sqrt{y_js_k}) \right]_{j,k=1}^n\prod_{j=1}^ne^{-\varepsilon s_j}
	.
	\end{split}
	\end{equation}
	To evaluate the $s$-integral, we push the regularisation in one of the determinants and then apply Andreief's formula to obtain the following.
	\begin{equation}
	\begin{split}
	&\int_{\R^n_+}\dv s \det\left[J_\nu(2\sqrt{y_js_k}) \right]_{j,k=1}^n\det\left[J_\nu(\sqrt{x_js_k})e^{-\varepsilon s_j} \right]_{j,k=1}^n
	\\
	=&n!\det\left[\int_{\R_+}\dv sJ_\nu(2\sqrt{x_js})J_\nu(2\sqrt{y_ks})e^{-\varepsilon s}
	\right]_{j,k=1}^n.
	\end{split}
	\end{equation}
	
	This new determinant can be again understood via a Berezin-Karpelevich integral~\eqref{2.2.6} (see also~\cite[Thm. 2.13 Case (3)]{FKK17}) so that we have
	\begin{equation}
	\begin{split}
	\s^{-1} [\s f](x)=&\frac{1}{n!C_\nu}\lim_{\varepsilon\rightarrow 0^+}\int_{\R^n_+} \dv y \,f(y)\Delta^2(x)\left(\prod_{j=1}^nx_j^{\nu}\right)\\
	&\times\frac{1}{\epsilon^{n(n+\nu)}}\int_{K_\nu}d\mu(K)\exp\left[-\frac{1}{2\varepsilon}\tr(\iota(\sqrt{x})-K\iota(\sqrt{y})K^{-1})^2\right]\\
	=&\frac{1}{n!C_\nu}\Delta^2(x)\left(\prod_{j=1}^nx_j^{\nu}\right)\lim_{\varepsilon\rightarrow 0^+}
	\frac{1}{\epsilon^{n(n+\nu)}}\int_{M_\nu}dX F(X)\exp\left[-\frac{1}{2\varepsilon}\tr(\iota(\sqrt{x})-X)^2\right].
	\end{split}
	\end{equation}
	To derive this result one has to note that $\tr(\iota(\sqrt{x}))^2=2\sum_{j=1}^nx_j$ for all matrix spaces $M_\nu$ and identify the joint probability density $f$ of the squared singular values with a probability density $F$ on $M_\nu$, see~\eqref{1.1.10}. Finally we shift $X$ by $\iota(\sqrt{x})$ and afterwards rescale it by $\varepsilon$. As $F\in L^1(M_\nu)$, the limit $\varepsilon\to0$ can be carried out for almost all $x\in\mathbb{R}_+^n$. Thus we arrive at
	\begin{equation}
	\s^{-1} [\s f](x)=\frac{\pi^{n(n+\nu)}}{n!C_\nu}\Delta^2(x)\left(\prod_{j=1}^nx_j^{\nu}\right) F(\iota(x))=f(x).
	\end{equation}
	This closes the proof.
\end{proof}

\section{Parametrising unitary matrices}\label{ab}

In 1897 Hurwitz~\cite{Hu97} gave a full parametrisation of the unitary groups by generalising the well-known Euler angles for $\SO(3)$ matrices. This parametrisation is summarised in~\cite[\S 2.3.1]{Fo10} and~\cite{DF17}. We will introduce a similar parametrisation for $\U(n)$, under which the diagonal entries of the $U$ matrix in its LU decomposition are factorised with the other parameters, described in proposition~\ref{pb.4}.

\begin{proposition}[Parametrisation for $\U(n)$]\label{pb.1}
	An $\U(n)$ matrix $V_n$ with $n>1$ has an iterative parametrisation specified by
	\begin{equation}\label{b.1}
	V_{n}=\begin{bmatrix}
	V_{n-1}&\\&1
	\end{bmatrix}H_n,\quad H_n=\Phi_{1,n}\Phi_{2,n}\ldots\Phi_{n-1,n},
	\end{equation}
	where $V_{n-1}\in\U(n)$. The $n\times n$ matrix $\Phi_{j,n}$ is given by
	\begin{equation}\label{b.2}
	\Phi_{j,n}=\begin{bmatrix}
	I_{j-1}	&													\\
	&\cos\phi_{j,n}	&			&	e^{i\psi_{j,n}}\sin\phi_{j,n}\\
	&					&I_{n-j-1}	&					\\
	&- e^{-i\psi_{j,n}}\sin\phi_{j,n}	&			&\cos\phi_{j,n}
	\end{bmatrix}
	\end{equation}
	 for $1<j<n$
	and
	\begin{equation}
	\Phi_{1,n}=\begin{bmatrix}
	e^{i\alpha_{n}}\cos\phi_{1,n} 	&			&e^{i\psi_{1,n}}\sin\phi_{1,n}	\\
	&I_{n-2}	&					\\
	-e^{-i\psi_{1,n}}\sin\phi_{1,n} 	&			&e^{-i\alpha_{n}}\cos\phi_{1,n}
	\end{bmatrix}.
	\end{equation}
	All empty entries are filled with $0$ and $I_{j}$ is the $j$th dimensional identity matrix.
	The non-zero entries comprising the angles $\phi_{j,n}$ are at positions $(j,j), (j,n), (n,j)$ and $(n,n)$. The ranges of the angles $\alpha_n,\psi_{j,n}$ and $\phi_{j,n}$ are given by
	\begin{equation}\label{b.3}
	-\pi\le \alpha_{n},\psi_{j,n}< \pi,\quad 0\le \phi_{j,n}\leq\pi/2,\quad (1\le j< n).
	\end{equation}
	For $n=1$, one can choose the standard parametrisation $V_1=e^{i\alpha_1}$ with $-\pi\leq\alpha_1< \pi$.
\end{proposition}

\begin{proof}
	
	Denote $V_{j,k}$ the $(j,k)$-th entry of the unitary matrix $V_n$. The proposition essentially states that we can iteratively reduce the last row of $V_n$ from the generic form $\vec v_n^{\,(1)}=(V_{n,1},\ldots,V_{n,n-1},V_{n,n}^{(1)})$ to $\vec v_n^{\,(2)}=(V_{n,1},\ldots,V_{n,n-2},0,V_{n,n}^{(2)})$, then to $\vec v_n^{\,(3)}=(V_{n,1},\ldots,V_{n,n-3},0,0,V_{n,n}^{(3)})$ and so forth till we arrive at $\vec v_n^{\,(n)}=(0,\ldots,0,V_{n,n}^{(n)})$ where $V_{n,n}^{(1)}=V_{n,n}$ and $V_{n,n}^{(n)}=1$. For $n=4$, one can sketch it in the following diagram:
	\begin{equation}
	V_4=\begin{bmatrix}
	*&*&*&*\\
	*&*&*&*\\
	*&*&*&*\\
	*&*&*&*
	\end{bmatrix}\xrightarrow{\displaystyle\Phi_{3,4}^\dagger}\begin{bmatrix}
	*&*&*&*\\
	*&*&*&*\\
	*&*&*&*\\
	*&*&0&*
	\end{bmatrix}\xrightarrow{\displaystyle\Phi_{2,4}^\dagger}\begin{bmatrix}
	*&*&*&*\\
	*&*&*&*\\
	*&*&*&*\\
	*&0&0&*
	\end{bmatrix}\xrightarrow{\displaystyle\Phi_{1,4}^\dagger}\begin{bmatrix}
	*&*&*&*\\
	*&*&*&*\\
	*&*&*&*\\
	0&0&0&*
	\end{bmatrix}.
	\end{equation}
	Here, we have to take into account that $\vec v_n^{\,(1)}$ describes a $(2n-1)$-dimensional sphere, i.e., $\vec v_n^{\,(1)}(\vec v_n^{\,(1)})^\dagger=1$. Since the matrices $\Phi_{j,n}$ shall be unitary themselves also the new vectors $\vec v_n^{\,(j)}$ with $j=1,\ldots,n-1$ are normalised though they only describe $(2n-2j+1)$-dimensional spheres.
	
	 The idea is to consider the two dimensional complex vectors $(V_{n,n-j},V^{(j)})\in\mathbb{C}^{2}$ which describes a real four dimensional unit ball because of the normalisation $\vec v_n^{\,(j)}(\vec v_n^{\,(j)})^\dagger=1$. Let us start with $1<j<n-1$. The non-trivial part of the matrix $\Phi_{n-j,n}$ reduces this four-dimensional ball to a two-dimensional one given by $V^{(j+1)}$, especially, we have
	 \begin{equation}
	 (V_{n,n-j},V^{(j)})=(0,V^{(j+1)})\left[\begin{array}{cc} \cos\phi_{n-j,n} & e^{i\psi_{n-j,n}} \sin\phi_{n-j,n} \\ -e^{-i\psi_{n-j,n}} \sin\phi_{n-j,n} & \cos\phi_{n-j,n} \end{array}\right].
	 \end{equation}
	 Explicitly this means
	 \begin{equation}\label{V-para}
	 V_{n,n-j}=- e^{-i\psi_{n-j,n}}\sin\phi_{n-j,n}\ V^{(j+1)}\quad {\rm and}\quad V^{(j)}= \cos\phi_{n-j,n}\ V^{(j+1)}.
	 \end{equation}
	The norm of $(V_{n,n-j},V^{(j)})$ carries over to the amplitude of $V^{(j+1)}$. We are even able to fix the complex phase of $V^{(j+1)}$ by choosing it the same of $V^{(j)}$ so that the angle $\phi_{n-j,n}$ only runs over the interval $0\leq \phi_{n-j,n}\leq \pi/2$. The independent complex phase of $V_{n,n-j}$ is taken care of by the angle $\psi_{n-j,n}$ and since it can be any value on the complex unit circle also the angle $\psi_{n-j,n}$ takes any value from $-\pi$ to $\pi$. The embedding of this transformation in the $n\times n$-dimensional matrix $\Phi_{n-j,n}\in\U(n)$ becomes clear on which columns this unitary matrix has to act and that the action on the other columns has to be via the identity as they have to be kept unchanged. Although $\Phi_{n-j,n}$ also acts on the other rows and not only on the last one, its action can be absorbed there as their orthogonality conditions with the last line as well as with each other are not affected.
	
	In the last step when $j=1$, the vector $(V_{n,1},V^{(n-1)})\in\mathbb{C}^{2}$ is essentially a three-dimensional sphere. It is well known that it can be parametrised by $(-e^{-i\psi_{1,n}}\sin\phi_{1,n},^{-i\alpha_{n}}\cos\phi_{1,n})$ with $-\pi\leq \psi_{1,n},\alpha_{n}<\pi$ describing the two independent complex phases and $0\leq\phi_{1,n}\leq\pi/2$ which give the amplitude of the two components of this vector. This two dimensional vector are the non-trivial components of the last row of $\Phi_{1,n}\in\U(n)$. The argument why $\Phi_{1,n}$ can be absorbed in the other rows is the same as for the case $j>1$.
	
	Once the last row of $V_n$ has been transformed to $(0,\ldots,0,1)$ the unitarity of the remaining matrix $V_nH_N^\dagger$, particularly the orthogonality of the last row with the others, implies that also the last column of $V_nH_N^\dagger$ is zero everywhere except for its last entry. This finishes the proof.
\end{proof}

As an example, we would like to explicitly present the case $V_3\in\U(3)$. Our decomposition reads
\begin{equation}
\begin{split}
V_3&=\begin{bmatrix}
e^{i\alpha_{1}}	& 0&0
\\
0&1&0
\\
0&0&1
\end{bmatrix}\begin{bmatrix}
\cos\phi_{1,2} e^{i\alpha_{2}}	& \sin\phi_{1,2}e^{i\psi_{1,2}}&0
\\
-\sin\phi_{1,2} e^{-i\psi_{1,2}}&\cos\phi_{1,2}e^{-i\alpha_{2}}&0
\\
0&0&1
\end{bmatrix}
H_3,
\\
H_3&=\begin{bmatrix}
\cos\phi_{1,3} e^{i\alpha_{3}}	& 0& \sin\phi_{1,3}e^{i\psi_{1,3}}
\\
0&1&0
\\
-\sin\phi_{1,3} e^{-i\psi_{1,3}}& 0&\cos\phi_{1,3}e^{-i\alpha_{3}}
\end{bmatrix}
\begin{bmatrix}
1&0&0
\\
0&\cos\phi_{2,3}	& \sin\phi_{2,3}e^{i\psi_{2,3}}
\\
0&-\sin\phi_{2,3} e^{-i\psi_{2,3}}&\cos\phi_{2,3}
\end{bmatrix}.
\end{split}
\end{equation}
One can check that $H_3$ has an upper triangular $2\times 2$ top left block, compatible with Lemma~\ref{pb.2}. This is the essential property of this parametrisation, allowing us to compute its LU decomposition using Proposition~\ref{pb.4}

It is also important to compare our parametrisation with Hurwitz' parametrisation, which is based on the following Euler's parametrisation for $\U(3)$:
\begin{equation}
\begin{split}
V_3&=
\begin{bmatrix}
e^{i\alpha_{1}}	& 0&0
\\
0&1&0
\\
0&0&1
\end{bmatrix}\begin{bmatrix}
\cos\phi_{1,2} e^{i\alpha_{2}}	& \sin\phi_{1,2}e^{i\psi_{1,2}}&0
\\
-\sin\phi_{1,2} e^{-i\psi_{1,2}}&\cos\phi_{1,2}e^{-i\alpha_{2}}&0
\\
0&0&1
\end{bmatrix}
E_3,
\\
E_3&=\begin{bmatrix}
1&0&0
\\
0&\cos\phi_{2,3}	& \sin\phi_{2,3}e^{i\psi_{2,3}}
\\
0&-\sin\phi_{2,3} e^{-i\psi_{2,3}}&\cos\phi_{2,3}
\end{bmatrix}
\begin{bmatrix}
\cos\phi_{1,3} e^{i\alpha_{3}}	& \sin\phi_{1,3}e^{i\psi_{1,2}}&0
\\
-\sin\phi_{1,3} e^{-i\psi_{1,2}}&\cos\phi_{1,3}e^{-i\alpha_{3}}&0
\\
0&0&1
\end{bmatrix}.
\end{split}
\end{equation}

The main difference between its Hurwitz' generalisation and our parametrisation is that the angles are in the block given by the $j$-th and $(j+1)$-th rows and columns in Hurwitz' case. In our choice those angles are associated with the $j$-th and $k$-th rows and columns instead.

We carry out the product $H_n=\Phi_{1,n}\Phi_{2,n}\ldots\Phi_{n-1,n}$, and the next lemma states what the matrix entries of $H_n$ explicitly are. In particular, the top-left $(n-1)\times(n-1)$ block of $H_n$ has an upper triangular form.

\begin{lemma}\label{pb.2}
	For $n>1$, let $H_n$ be defined as in~\eqref{2.3.1}. Its matrix entries $h_{j,k}^{(n)}$ with $j,k=1,\ldots,n$ are
	\begin{equation}\label{b.7a}
	h_{j,k}^{(n)}=\begin{cases}
	0, &k<j<n,\\
	e^{i\alpha_n}\cos\phi_{1,n} , &1=j=k,\\
	\cos\phi_{j,n},& 1<j=k<n,\\
	-e^{i(\psi_{j,n}-\psi_{k,n})}\sin\phi_{k,n}\sin\phi_{j,n}\prod_{l=j+1}^{k-1}\cos\phi_{l,n}, &j<k<n,\\
	e^{i\psi_{j,n}}\sin\phi_{j,n}\prod_{l=j+1}^{n-1}\cos\phi_{l,n}, &j<k=n,
	\end{cases}
	\end{equation}
	\begin{equation}\label{b.7}
	h_{n,k}^{(n)}=\begin{cases}
	-e^{-i\psi_{1,n}}\sin\phi_{1,n},& k=1,\\
	-e^{-i(\psi_{k,n}+\alpha_n)}\sin\phi_{k,n}\prod_{l=1}^{k-1}\cos\phi_{l,n},&1<k<n,\\
	e^{-i\alpha_n}\prod_{l=1}^{n-1}\cos\phi_{l,n},&k=n.
	\end{cases}
	\end{equation}
\end{lemma}

\begin{proof}
	The statement for the last row $(h_{n,1}^{(n)},\ldots,h_{n,n}^{(n)})$ follows from the iteration~\eqref{V-para}. Let $H_n^{(m)}=\Phi_{1,n}\Phi_{2,n}\ldots\Phi_{m,n}=\{h_{j,k}^{(m)}\}_{j,k=1,\ldots,n}$ for $m\geq1$. Then, we have indeed the recursion
	\begin{equation}
	\begin{split}
	h_{n,k}^{(m+1)}=&h_{n,k}^{(m)},\qquad {\rm for}\ k\neq m+1,n,\\
	h_{n,m+1}^{(m+1)}=&-e^{i\psi_{m+1,n}}\sin\phi_{m+1,n}\ h_{n,n}^{(m)},\\
	h_{n,n}^{(m+1)}=&\cos\phi_{m+1,n}\ h_{n,n}^{(m)},
	\end{split}
	\end{equation}
	as the initial condition is $(h_{n,1}^{(1)},\ldots,h_{n,n}^{(1)})=(-e^{-i\psi_{1,n}}\sin \phi_{1,n},0,\ldots,0,e^{-i\alpha_{n}}\cos \phi_{1,n})$. In this way, the final row will be equal to
	\begin{equation}
	(h_{n,1}^{(n)},\ldots,h_{n,n}^{(n)})=(-e^{-i\psi_{1,n}}\sin \phi_{1,n},-e^{i\psi_{2,n}}\sin\phi_{2,n}\ h_{n,n}^{(2)},\ldots,-e^{i\psi_{n,n}}\sin\phi_{n,n}\ h_{n,n}^{(n-1)},h_{n,n}^{(n)}).
	\end{equation}
	The recursion for $h_{n,n}^{(m)}$ can be explicitly expressed in terms of the product shown in~\eqref{b.7}.
	
	A similar computation can be performed for the other rows of $H_n$. For this purpose, it is helpful to notice that
	the $j$th row stays the Kronecker symbol $h_{j,k}^{(l)}=\delta_{j,k}$ as long as $l<j$, as the matrices $\Phi_{l,n}$ act like the identity on this vector. Once $l=j$, we can set the row vector equal to the starting initial condition which is $(h_{1,1}^{(1)},\ldots,h_{1,n}^{(1)})=(e^{i\alpha_{n}}\cos \phi_{1,n},0,\ldots,0,e^{i\psi_{1,n}}\sin \phi_{1,n})$ for $j=1$ and $(h_{j,1}^{(1)},\ldots,h_{j,n}^{(1)})=(\cos \phi_{j,n},0,\ldots,0,e^{i\psi_{j,n}}\sin \phi_{j,n})$ for $1<j<n$. Then, we can write the recursion relations
	\begin{equation}
	\begin{split}
	h_{j,k}^{(m+1)}=&h_{j,k}^{(m)},\qquad {\rm for}\ k\neq m+1,n,\\
	h_{j,m+1}^{(m+1)}=&-e^{i\psi_{m+1,n}}\sin\phi_{m+1,n}\ h_{j,n}^{(m)},\\
	h_{j,n}^{(m+1)}=&\cos\phi_{m+1,n}\ h_{j,n}^{(m)},
	\end{split}
	\end{equation}
	for $m\geq j$ for the $j$th row of $H_n^{(m)}$. This means that $h_{j,k}^{(m)}$ stays unchanged for $k<j<n$ and remains zero the whole time. The recursions can be again traced back to the one for $h_{j,n}^{(m)}$. This yields~\eqref{b.7a} after resolving the one which anew is a simple product.
\end{proof}

For explicitly evaluating the spherical transform on $\U(N)$, we need the LU-decomposition which we give now in the new coordinates that have been developed. The ensuing proposition makes essential use of the fact, implied by Lemma~\ref{pb.2}, that the $(n-1)\times (n-1)$ block of $H_n$ is upper triangular.

\begin{proposition}[LU decomposition for $\U(n)$]\label{pb.4}
	Let $V_n\in\U(n)$ and has the LU decomposition $V_n=L_nU_n.$ where the matrix entries of $U_n\in\U(n)$ are denoted by $u_{j,k}^{(n)}$. Then, the diagonal element is given by
	\begin{equation}\label{b.16}
	u_{l,l}^{(n)}=\begin{cases}
	\displaystyle\exp\left[i\sum_{j=1}^n\alpha_{j}\right]\prod_{k=2}^{n}\cos\phi_{1,k}, &l=1,\\\\
	\displaystyle e^{-i\alpha_{l}}\frac{\prod_{k=l+1}^{n}\cos\phi_{l,k} }{\prod_{j=1}^{l-1}\cos\phi_{j,l}}, &l\ne 1,
	\end{cases}
	\end{equation} 
	 in the iterative coordinates introduced in Proposition~\ref{pb.1}.
	For the radius $r_{l}^{(n)}$ and complex phase $e^{i\varphi_l^{(n)}}$ of the product $u_{1,1}^{(n)}u_{2,2}^{(n)}\cdots u_{l,l}^{(n)}$, we find
	\begin{equation}\label{b.22}
	r_l^{(n)}=\prod_{j=1}^l\prod_{k=l+1}^n\cos\phi_{j,k},\quad \exp[i\varphi_l^{(n)}]=\exp\left[i\alpha_1+i\sum_{j=l+1}^{n}\alpha_j\right],\quad (l=1,\ldots,n-1).
	\end{equation}
	Also, the determinant $\exp[i\varphi_n^{(n)}]$ of $V_n$ is given by $\exp[i\varphi_n^{(n)}]=\exp[i\alpha_1]$.
\end{proposition}

\begin{proof}
	
	Equation~\eqref{b.22} is a simple consequence of~\eqref{b.16}. Thus, we concentrate ourselves on proving the latter.

	We obtain the expression for $u_{l,l}^{(n)}$ using a recursive procedure for $n>1$ as for $n=1$ we trivially have $u_{1,1}^{(1)}=e^{i\alpha_1}$. To this aim, we decompose $V_{n-1}$ into $L_{n-1}$ and $U_{n-1}$ as follows
	\begin{equation}
	Q_n
	=\begin{bmatrix}
	L_{n-1}&0\\0&1
	\end{bmatrix}\begin{bmatrix}
	U_{n-1}&0\\0&1
	\end{bmatrix}H_n.
	\end{equation}
	Now we decompose the product of $\begin{bmatrix}
	U_{n-1}&0\\0&1
	\end{bmatrix}$ and $H_n$. This is best done in a block-wise multiplication and decomposing it afterwards, i.e.,
	\begin{equation}
	\begin{bmatrix}
	U_{n-1}&0\\0&1
	\end{bmatrix}
	\begin{bmatrix}
	[h_{j,k}^{(n)}]_{j,k=1}^{n-1}&[h_{j,n}^{(n)}]_{j=1}^{n-1}\\ \\
	[h_{n,k}^{(n)}]_{k=1}^{n-1}&h_{n,n}^{(n)}
	\end{bmatrix}=\begin{bmatrix}
	U_{n-1}[h_{j,k}^{(n)}]_{j,k=1}^{n-1}&U_{n-1}[h_{j,n}^{(n)}]_{j=1}^{n-1}\\ \\
	[h_{n,k}^{(n)}]_{k=1}^{n-1}&h_{n,n}^{(n)}
	\end{bmatrix}=L'U_{n},
	\end{equation}
	where $L'$ is a lower triangular matrix, obviously satisfying $L_{n}=\diag(L_{n-1},1)L'$ because we aim at the decomposition $V_n=L_nU_n$.
	By Lemma~\ref{pb.2}, $[h_{j,k}^{(n)}]_{j,k=1}^{n-1}$ is upper triangular, which implies that the first $(n-1)$ diagonal entries of $U_n$ are products of the corresponding diagonal entries of $U_{n-1}$ and $H_n$, i.e.
	\begin{equation}
	u_{l,l}^{(n)}=u_{l,l}^{(n-1)}h_{l,l}^{(n)},\quad {\rm for}\ l=1,\cdots,n-1\ {\rm and}\ n>1.
	\end{equation}
	Also $u_{n,n}^{(n)}$ is obtained from the determinant requirement $u_{1,1}^{(n)}u_{2,2}^{(n)}\cdots u_{n,n}^{(n)}=e^{i\alpha_1}$ as apart from $V_1$ all the other matrices $\Phi_{j,k}$ have the determinant equal to unity.
\end{proof}

Eventually, we would like to give the Haar measure of $\U(n)$ in these new coordinates.

\begin{proposition}[Haar measure for $\U(n)$]\label{prop:Haar.coord}
	In the parametrisation of Proposition~\ref{pb.1}, the normalised Haar measure of $V_n\in\U(n)$ reads
	\begin{equation}\label{b.8a}
	\mu(\dv V_n)=\left(\prod_{1\le j<k\le n}2(k-j)(\cos\phi_{j,k})^{2(k-j)-1}\sin\phi_{j,k}\dv\phi_{j,k}\right)\left(\prod_{j=1}^n\frac{\dv\alpha_j}{2\pi}\right)\left(\prod_{1\le j<k\le n}\frac{\dv\psi_{j,k}}{2\pi}\right).
	\end{equation}
\end{proposition}

\begin{proof}
	We follow the proof in~\cite[\S 2.3.1]{Fo10} and compute the Jacobian recursively, too. In particular, we adapt the formula presented therein which is
	\begin{equation}\label{b.8}
	\mu(\dv V_n)=C_n\mu(\dv V_{n-1})\left|\det\left[\begin{array}{cc|c}
	\displaystyle\left\{\frac{\partial \vec h_{n}^{(n)}}{\partial \alpha_{n}}H_n^\dagger\right\}_{j,r}
	&\displaystyle\left\{\frac{\partial \vec h_{n}^{(n)}}{\partial \alpha_{n}}H_n^\dagger\right\}_{j,i}
	&\displaystyle\left\{\frac{\partial \vec h_{n}^{(n)}}{\partial \alpha_{n}}H_n^\dagger\right\}_{n,i}\\ \hline
	\displaystyle\left\{\frac{\partial \vec h_{n}^{(n)}}{\partial \phi_{k,n}}H_n^\dagger\right\}_{j,r}
	&\displaystyle\left\{\frac{\partial \vec h_{n}^{(n)}}{\partial \phi_{k,n}}H_n^\dagger\right\}_{j,i}
	&\displaystyle\left\{\frac{\partial \vec h_{n}^{(n)}}{\partial \phi_{k,n}}H_n^\dagger\right\}_{n,i}\\
	\displaystyle\left\{\frac{\partial \vec h_{n}^{(n)}}{\partial \psi_{k,n}}H_n^\dagger\right\}_{j,r}
	&\displaystyle\left\{\frac{\partial \vec h_{n}^{(n)}}{\partial \psi_{k,n}}H_n^\dagger\right\}_{j,i}
	&\displaystyle\left\{\frac{\partial \vec h_{n}^{(n)}}{\partial \psi_{k,n}}H_n^\dagger\right\}_{n,i}
	\end{array}\right]_{j,k=1,\ldots,n-1}\right|,
	\end{equation}
	where $\vec h_{j}^{(n)}$ for the $j$th row of $H_n$ and $\{\}_{j,r}$ and $\{\}_{j,i}$ specifies the real and imaginary part of the $j$-th entry. The size of the matrix in~\eqref{b.8} is $(2n-1)\times (2n-1)$ since the last entry of $\dv \vec h_n^{(n)} H_n^\dagger$ is purely imaginary. The factorisation of the measures is not surprising as it reflects the group factorisation $\U(n)=\U(n-1)\times [\U(n)/\U(n-1)]$. The coset $\U(n)/\U(n-1)$ is the $(2n-1)$-dimensional sphere which is parametrised by the last row of $H_n$. Finally, let us underline that the angles $\alpha_n$, $\psi_{k,n}$ and $\psi_{k,n}$ are only comprised in $H_n$ and $V_{n-1}$ is independent of those.
	
	To evaluate the determinant in~\eqref{b.8}, we work out the invariant length element of the vector $\vec h_{n}^{(n)}$ as the determinant of the corresponding Riemannian metric is proportional to the square of this determinant. This computation becomes simpler when writing $\vec h_{n}^{(n)}=\Psi\vec\chi_n$ with the diagonal matrix of complex phases $\Psi=\diag(-e^{-i\psi_{1,n}},-e^{-i(\psi_{2,n}+\alpha_n)},\ldots,-e^{-i(\psi_{n-1,n}+\alpha_n)},e^{-i\alpha_{n}})$ and the remaining real vector $\vec\chi_n$. This real vector describes an $(n-1)$-dimensional real sphere, in particular it is $\vec\chi_n\vec\chi_n^{\,T}=1$, and admits the recursion
	\begin{equation}
	\vec\chi_j=(\sin\phi_{n-j+1,n},\ \cos\phi_{n-j+1,n}\,\vec\chi_{j-1})
	\end{equation}
	for all $j=2,\ldots,n$ with $\vec\chi_{1}=1$. With the help of this convention, we compute the length element
	\begin{equation}
	\begin{split}
	\dv\vec h_{n}^{(n)}\dv(\vec h_{n}^{(n)})^\dagger=&\vec\chi_n(\dv\Psi)(\dv\Psi^\dagger)\vec\chi_n^{\,T}+\dv\vec\chi_n\dv\vec\chi_n^{\,T}\\
	=&\sin^2\phi_{1,n}\,\dv\psi_{1,n}^2+\sum_{k=2}^{n-1}\sin^2\phi_{k,n}\left(\prod_{l=1}^{k-1}\cos^2\phi_{l,n}\right)(\dv \psi_{k,n}+\dv\alpha_n)^2\\
	&+\left(\prod_{l=1}^{n-1}\cos^2\phi_{l,n}\right)\dv\alpha_n^2+\dv\phi_{1,n}^2+\cos^2\phi_{1,n}\dv\vec\chi_{n-1}\dv\vec\chi_{n-1}^{\,T}\\
	=&\sin^2\phi_{1,n}\,\dv\psi_{1,n}^2+\sum_{k=2}^{n-1}\sin^2\phi_{k,n}\left(\prod_{l=1}^{k-1}\cos^2\phi_{l,n}\right)(\dv \psi_{k,n}+\dv\alpha_n)^2\\
	&+\left(\prod_{l=1}^{n-1}\cos^2\phi_{l,n}\right)\dv\alpha_n^2+\dv\phi_{1,n}^2+\sum_{k=2}^{n-1}\left(\prod_{l=1}^{k-1}\cos^2\phi_{l,n}\right)\dv\phi_{k,n}^2.
	\end{split}
	\end{equation}
	In the first line, the mixed terms vanish as $\vec\chi_n(\dv\Psi)\Psi^\dagger d\vec\chi_n^{\,T}=-d\vec\chi_n\Psi( \dv\Psi^\dagger)\vec\chi_n^{\,T}$. In the second step, we have made used of the fact that also $\vec\chi_{n-1}$ is a real unit vector so that we have $\vec\chi_{n-1}\dv\vec\chi_{n-1}^{\,T}=0$. 
	
	 Despite that we encounter in the invariant length element the combination $\dv \psi_{k,n}+\dv\alpha_n$, the determinant of the corresponding metric is up to a numerical factor equal to the product
	 \begin{equation}
	 \begin{split}
	 &\sin^2\phi_{1,n}\left[\prod_{k=2}^{n-1}\sin^2\phi_{k,n}\left(\prod_{l=1}^{k-1}\cos^2\phi_{l,n}\right)\right]\left(\prod_{l=1}^{n-1}\cos^2\phi_{l,n}\right)\left[\prod_{k=2}^{n-1}\left(\prod_{l=1}^{k-1}\cos^2\phi_{l,n}\right)\right]\\
	 =&\prod_{j=1}^{n-1}\sin^2\phi_{j,n}(\cos\phi_{j,n})^{2(2n-2j-1)}.
	 \end{split}
	 \end{equation}
	 After we take the root, we obtain the factor in~\eqref{b.8a} that depends on $\phi_{j,n}$. Resolving the recursion from $V_{n-1}\in\U(n-1)$ to $V_{1}\in\U(1)$ we get the remaining terms in the same way.
	
	The normalisation can be readily computed as all angles are independent.
\end{proof}

Let us emphasise that the radii and complex phases in~\eqref{b.22} are the only parts of $V_n\in\U(n)$ which enter the spherical transform. All other variables do not play an important role and need to be integrated out. This motivates us to rewrite the Haar measure even further in these variables.

\begin{corollary}
	For $n>1$, the Haar measure~\eqref{b.8a} can be rewritten as
	\begin{equation}\label{b.18}
	\mu(\dv V_n)=\left(\prod_{k=1}^n\frac{(k-1)!}{2\pi^k}\right)\dv\alpha_1\prod_{l=1}^{n-1}r_l^{(n)}\dv r_l^{(n)}\prod_{l=1}^{n-1}\dv\varphi_l^{(n)} \prod_{\substack{2\leq j+1<k\leq n}}\tan\phi_{j,k}\dv\phi_{j,k}\prod_{1\leq j<k\leq n}\dv\psi_{j,k},
	\end{equation}
	where $-\pi\le \varphi_l^{(n)}<\pi$, the ranges of $r_l^{(n)}$ are specified by
	\begin{equation}\label{b.18b}
	0<r_j^{(n)}<R_j=\frac{\prod_{l=1}^j\prod_{k=j+1}^n\cos\phi_{l,k}}{\cos\phi_{j,j+1}}\leq1,\quad(j=1,\ldots, n-1),
	\end{equation}
	and the ranges of the other angles are the same as in~\eqref{b.3}.
\end{corollary}
\begin{proof}
	The inequality~\eqref{b.18b} can be obtained by using Corollary~\ref{pb.2} and changing the variables $\phi_{j,j+1}$ to $r_j^{(n)}$ and $\alpha_{j+1}$ to $\varphi_j^{(n)}$ with $j=1,\ldots,n-1$. The computation of the Jacobian of the change of variables from $\alpha_{j+1}$ to $\varphi_j^{(n)}$ is equal to unity because they end up in trivial shifts due to $\varphi_{l}^{(n)}=\varphi_{l+1}^{(n)}+\alpha_{l+1}$. The Jacobian for the change from $\phi_{l,l+1}$ to $r_l^{(n)}$ gives $\dv r_l^{(n)}/r_l^{(n)}=\tan \phi_{j,j+1}\dv \phi_{j,j+1}$ which has to be done successively starting with $l=1$. The product in the measure~\eqref{b.8a} is in these coordinates
	\begin{equation}
	\prod_{k=2}^n\prod_{j=1}^{k-1}\sin\phi_{j,k}(\cos\phi_{j,n})^{2k-2j-1}=\left(\prod_{k=2}^n (r_{k-1}^{(n)})^2\right)\left(\prod_{k=2}^n\prod_{j=1}^{k-1}\tan\phi_{j,k}\right),
	\end{equation}
	 which eventually shows~\eqref{b.18}. 
\end{proof}


\end{document}